\documentclass[conference]{IEEEtran}

\ifCLASSINFOpdf
\else
\fi

 \usepackage[caption=false,font=footnotesize]{subfig}

\usepackage{url}

\usepackage{amsmath}
\usepackage{amssymb}
\usepackage{xcolor}
\usepackage[most]{tcolorbox}
\usepackage{fontawesome5}
\usepackage{ifthen}
\usepackage{pifont}
\usepackage{bm}
\usepackage{placeins}
\usepackage{array}
\usepackage{booktabs}
\usepackage{tabularx}
\usepackage{mdframed}

\newcommand{\arrowcloak}{\textit{ArrowCloak}}
\newcommand{\arrowmatch}{\textit{ArrowMatch}}

\newcommand{\attackname}{\textit{ArrowRevelio}}

\newcounter{observation}[section]
\renewcommand{\theobservation}{O\arabic{observation}}

\newcounter{guarantee}[section]
\renewcommand{\theguarantee}{G\arabic{guarantee}}

\newtheorem{lemma}{Lemma}
\newtheorem{theorem}{Theorem}

\newcommand{\Int}{\mathbb{Z}}

\newcommand{\Zmod}[1]{\mathbb{Z}_{#1}}

\newcommand{\Unif}{\mathcal{U}}                     %
\newcommand{\samples}{\leftarrow}                   %

\newcommand{\lwedim}{n}                             %
\newcommand{\lwesamp}{m}                            %
\newcommand{\lwemod}{q}                             %
\newcommand{\errdist}{\chi}                         %

\newcommand{\tp}{^{\top}}
\newcommand{\inv}{^{-1}}

\newcommand{\Id}{I}
\newcommand{\ones}{\bm{1}}
\newcommand{\zerovec}{\bm{0}}
\newcommand{\stdbasis}[1]{\bm{e}_{#1}}
\DeclareMathOperator{\diag}{diag}

\DeclareMathOperator{\spn}{span}

\DeclareMathOperator{\cossim}{cos}

\newcommand{\opnorm}[1]{\lVert #1 \rVert_{2}}
\newcommand{\frobnorm}[1]{\lVert #1 \rVert_{F}}

\newcommand{\round}[1]{\left\lfloor #1 \right\rceil}
\newcommand{\ip}[2]{\langle #1 , #2 \rangle}
\newcommand{\colof}[2]{#1_{:,#2}}                   %
\newcommand{\rowof}[2]{#1_{#2,:}}                   %

\definecolor{rqresultbg}{HTML}{E6E6E6}

\newenvironment{rqresultbox}[1]{%
  \begin{tcolorbox}[
    enhanced,
    breakable,
    colback=rqresultbg,
    colframe=black!75,
    boxrule=0.5pt,
    arc=0mm,
    outer arc=0mm,
    left=1.5mm,
    right=1.5mm,
    top=1.2mm,
    bottom=1.2mm,
    before skip=6pt,
    after skip=6pt
  ]
  \textbf{Results (RQ #1):}\ %
}{%
  \end{tcolorbox}
}

\newif\ifdraft %
\drafttrue

\newif\ifarxiv
\arxivtrue

\begin{document}
\title{Hiding Directions, Leaking Structure: Breaking ArrowCloak through Low-Rank Structure}

\ifarxiv
\IEEEoverridecommandlockouts
\author{\IEEEauthorblockN{Beijie Liu\IEEEauthorrefmark{1}\IEEEauthorrefmark{2},
Junyi Ouyang\IEEEauthorrefmark{1}\IEEEauthorrefmark{2},
Haoxuan Xu\IEEEauthorrefmark{1}\IEEEauthorrefmark{2},
Vincent Quentin Ulitzsch\IEEEauthorrefmark{3},\\
Potung Yu\IEEEauthorrefmark{2},
Yajie Zhao\IEEEauthorrefmark{2}, and
Mengyuan Li\IEEEauthorrefmark{2}}
\IEEEauthorblockA{\IEEEauthorrefmark{2}University of Southern California,
\{beijieli, junyiouy, xuhaoxua, potungyu, mli49061\}@usc.edu, zhao@ict.usc.edu}
\IEEEauthorblockA{\IEEEauthorrefmark{3}MIT, viniul@mit.edu}
\thanks{\textsuperscript{*} Equal contribution.}}
\else
\author{\IEEEauthorblockN{Anonymous Submission}}
\fi

\ifarxiv
\else
    \IEEEoverridecommandlockouts
    \makeatletter\def\@IEEEpubidpullup{6.5\baselineskip}\makeatother
    \IEEEpubid{\parbox{\columnwidth}{
    		Network and Distributed System Security (NDSS) Symposium 2027\\
    		22--26 March 2027, Seoul, Republic of Korea\\
    		ISBN 978-1-970672-09-1\\  
    		https://dx.doi.org/10.14722/ndss.2027.[23$|$24]xxxx\\
    		www.ndss-symposium.org
    }
    \hspace{\columnsep}\makebox[\columnwidth]{}}
\fi
\newcommand{\bheading}[1]{{\vspace{2pt}\noindent{\textbf{#1}}}}
 \maketitle

\begin{abstract}
TEE-shielded inference keeps sensitive state in a trusted execution environment
(TEE) while offloading linear algebra to an untrusted accelerator. Wang et al.,
in \textit{Game of Arrows} (USENIX Security 2025), showed that five widely
adopted lightweight defenses preserve vector directions and introduced
\arrowmatch{} to exploit this leakage. They then proposed \arrowcloak{}, which
adds a different multiple of one shared mask direction to each vector and bases
its weight-recovery hardness argument on Learning with Errors (LWE).
\arrowcloak{} successfully reduces \arrowmatch{} to near-black-box levels.

In this paper, we revisit \arrowcloak{} from cryptographic and structural
perspectives. Its LWE formulation does not by itself establish standard LWE
hardness: the reduction direction, quantized arithmetic, and joint instance
distribution do not meet the required conditions. Reusing one mask direction
thus leaves a recoverable rank-one component across the released matrix. We
exploit this structure with \attackname{}, an end-to-end, query-free recovery
attack. Given a public checkpoint and the obfuscated weights, \attackname{}
removes the masking subspace, recovers the hidden one-to-one correspondence,
and reconstructs protected weights without transformation secrets, victim
queries, or fine-tuning data. Across six model--task pairs spanning
classification, segmentation, and diffusion, \attackname{} recovers
$99.92$--$100\%$ of hidden vector correspondences. Reconstructed classification
models achieve $94.39$--$99.54\%$ victim agreement and differ by at most $1.59$
percentage points in accuracy; the recovered segmentation model achieves
$98.35\%$ output agreement. These findings suggest that lightweight protection
should address both per-vector geometry and joint structure across released
weights.
\end{abstract}

\IEEEpeerreviewmaketitle

\section{Introduction}

Deploying proprietary neural-network models on users' devices enables local,
low-latency inference, but it places valuable model weights on hardware
controlled by potential adversaries~\cite{sun2021mind,nayan2024sok}. Unlike a
cloud service that exposes only a query interface, on-device deployment gives
an adversary a white-box view of the software stack, memory, and hardware
interfaces, creating a direct risk of model extraction.

A natural defense is to protect model weights within a trusted execution
environment (TEE). However, running an entire modern model inside a CPU TEE
forgoes GPU-class acceleration, while GPU TEEs remain unavailable on much of
today's deployed hardware, with mature support largely confined to recent
NVIDIA server-grade accelerators~\cite{nvidia_h100_cc}. This gap has motivated split-inference systems that
keep sensitive state and lightweight protection logic inside a CPU TEE while
outsourcing linear algebra to an untrustworthy accelerator. 

Slalom~\cite{slalom}
pioneered this design by masking activations and verifying GPU computation.
Subsequent systems adapted this split to model confidentiality. Following prior
terminology, we call the LLM instance of this architecture \textit{TEE-Shielded LLM
Partitioning (TSLP)}~\cite{slalom,mo2020darknetz,hou2022modelprotection,
shen2022soter,sun2023shadownet,zhou2023nnsplitter,zhang2024noprivacy,
zhang2024groupcover,li2024translinkguard,xu2024tempo,yang2024kvshield,
sun2025tsqp,li2025coreguard,gameofarrows2025,xiong2025loro}.
These systems either shield selected model components or obfuscate outsourced
weights. For weight-obfuscation-based TSLP, keeping trusted-side work cheaper
than the matrix multiplication sent to the accelerator favors lightweight
per-vector transformations.

Wang et al., in \textit{Game of Arrows} (USENIX Security'25)~\cite{gameofarrows2025},
systematically examined the confidentiality consequences of this low-overhead
obfuscation strategy. They showed that these lightweight per-vector schemes,
widely adopted in
TSLP~\cite{shen2022soter,sun2023shadownet,li2024translinkguard,xu2024tempo,
yang2024kvshield,sun2025tsqp,li2025coreguard}, remain vulnerable to
direction-based recovery: their dominant scaling and permutation operations
change vector magnitudes and positions but preserve
directions. Their analysis assumes models fine-tuned from public checkpoints, with transformed weights exposed to an untrusted GPU while transformation secrets and recovery remain protected by the TEE.
Specifically, their \arrowmatch{} attack uses these directions to recover the hidden permutation
between public and exposed vectors, then uses query-labeled data to fit vector
lengths and construct a surrogate. To counter this
attack, the same work introduced \arrowcloak{}, a lightweight direction-hiding
defense. \arrowcloak{} scales each vector, adds a
vector-specific multiple of one shared mask direction, and permutes the result.
The additive component generally changes the exposed direction, disrupting the
signal used by \arrowmatch{}, while compact secret state lets the TEE correct
the GPU output. \arrowcloak{} further presents an  Learning with Errors (LWE)-based security argument
for the computational hardness of weight recovery. The authors report that it
effectively mitigates \arrowmatch{}, making obfuscated directions nearly as
dissimilar as random vectors.

In this paper, we revisit \arrowcloak{} along two complementary dimensions: whether its proposed LWE formulation actually establishes computational hardness, and whether its direction-hiding transformation eliminates exploitable structure in the released weights. 
We find weaknesses on both fronts, showing that \arrowcloak{} does not provide its intended model confidentiality. 
We first reassess \arrowcloak{}'s proposed connection to LWE. Rewriting an \arrowcloak{} weight relation in an LWE-like form does not transfer LWE hardness: the proposed mapping does not constitute a hardness-preserving reduction from LWE to \arrowcloak{} recovery, its quantized arithmetic fails to preserve the claimed modular relation, and its induced joint distribution over matrices, secrets, and errors does not match standard LWE~\cite{regev2009lwe,peikert2016decade}. These gaps invalidate the stated hardness argument, but do not by themselves imply that \arrowcloak{} can be efficiently broken. We therefore turn to the concrete transformation and ask whether its direction-hiding mechanism nevertheless leaves recoverable structure. We find that it does. Although \arrowcloak{} changes individual vector directions, its masking terms are not independent: they lie in one shared direction in the original construction, or in a shared low-dimensional subspace under the rank-$r$ extension we analyze. Under the conditions we derive, this shared structure can be estimated from the jointly released vectors and projected away, restoring enough of their relationship to the public initialization to recover the hidden permutation.

Based on this insight, we present \textbf{\attackname{}}, an efficient, query-free recovery attack against \arrowcloak{}-type protection. Given the public initialization checkpoint, public architecture metadata, and the obfuscated weights exposed to the GPU, \attackname{} uses singular value decomposition (SVD) to estimate and remove the shared masking subspace, solves a global one-to-one assignment to recover the hidden permutation, and fits the remaining coefficients to reconstruct the victim weights. The attack requires no plaintext victim weights, transformation secrets, victim-model queries, or fine-tuning data. We further derive sufficient spectral-separation conditions for recovering the original rank-one mask direction and its removal projector, and characterize the corresponding failure conditions for rank-$r$ masking. Across six model--task pairs, \attackname{} achieves $99.92$--$100\%$ permutation recovery. On ViT-B/16, BERT-base, and GPT-2 classification workloads, the reconstructed models agree with their victims on $94.39$--$99.54\%$ of test inputs, with task accuracy differing by at most $1.59$ percentage points. The recovered segmentation model achieves $98.35\%$ output agreement with its victim, and we additionally evaluate two diffusion workloads.

Viewed together, \textit{Game of Arrows} and our work expose complementary dimensions of leakage in TSLP's weight protection. \arrowmatch{} exploits geometry preserved within individual vectors under scaling and permutation, whereas \attackname{} exploits structure shared across vectors after \arrowcloak{} conceals this local geometry. This progression shows that preventing per-vector matching is insufficient when the jointly released weights still expose recoverable matrix-level structure. Confidential offloading mechanisms must therefore account for both local geometric leakage and global structure across released weights.

This observation raises a natural design question: can \arrowcloak{} mitigate cross-vector leakage by increasing the mask rank while preserving its lightweight trusted-side design? We study this tradeoff on BERT-base/SST-2 by increasing the mask rank from $64$ to $512$ and applying the complete attack with automatic rank estimation. Higher-rank masking progressively weakens functional recovery, but does not eliminate the underlying structural leakage: permutation recovery remains $100\%$ through $r=256$ and $99.98\%$ at $r=512$, while victim agreement falls from $95.53\%$ to $48.05\%$. These results show that higher-rank masking weakens the attack, but requires more trusted-side state and correction work without fully eliminating the underlying structural leakage.

To the best of our knowledge, the only other public critique of \arrowcloak{} is a concurrent work by Chiang et al., released as the \textit{MOSAIC} arXiv preprint on July 31, 2026~\cite{chiang2026mosaic}. Its Appendix~F identifies a flaw in \arrowcloak{}'s LWE-based security argument, but does not evaluate whether the scheme is practically recoverable. Our work audits the argument more broadly, covering its reduction direction, quantized arithmetic, and joint instance distribution. We also present \attackname{}, the first end-to-end, query-free attack to recover a functional model from \arrowcloak{}-protected weights.

This paper makes the following contributions:
\begin{itemize}
    \item We audit \arrowcloak{}'s LWE-based security argument and show that it does not transfer LWE hardness, identifying gaps in the reduction direction, quantized arithmetic, and joint instance distribution.

    \item We uncover shared low-rank masking structure in \arrowcloak{} and develop \attackname{}, a polynomial-time, query-free recovery attack using only public information and the exposed obfuscated weights. We also derive sufficient conditions for rank-one mask recovery.

    \item We evaluate \attackname{} across six model--task pairs, achieving $99.92$--$100\%$ permutation recovery and recovering victim functionality across classification, segmentation, and diffusion workloads.

    \item We study higher-rank masking as a mitigation and show that it weakens functional recovery but does not eliminate subspace or permutation recovery, while increasing trusted-side cost.
\end{itemize}

\section{TEE-Shielded LLM Partitioning}
\label{sec:background}

This section introduces the deployment setting and protection mechanism
analyzed in this paper. We first describe partitioned inference between a TEE
and an untrusted GPU, formalize \arrowcloak{}'s weight transformation and
efficient output recovery, and summarize its claimed security argument. 
Table~\ref{tab:notation} summarizes the notation conventions used throughout
the paper; problem-specific symbols are defined when they are first introduced.

\begin{figure*}[t]
    \centering
    \includegraphics[width=0.95\textwidth]{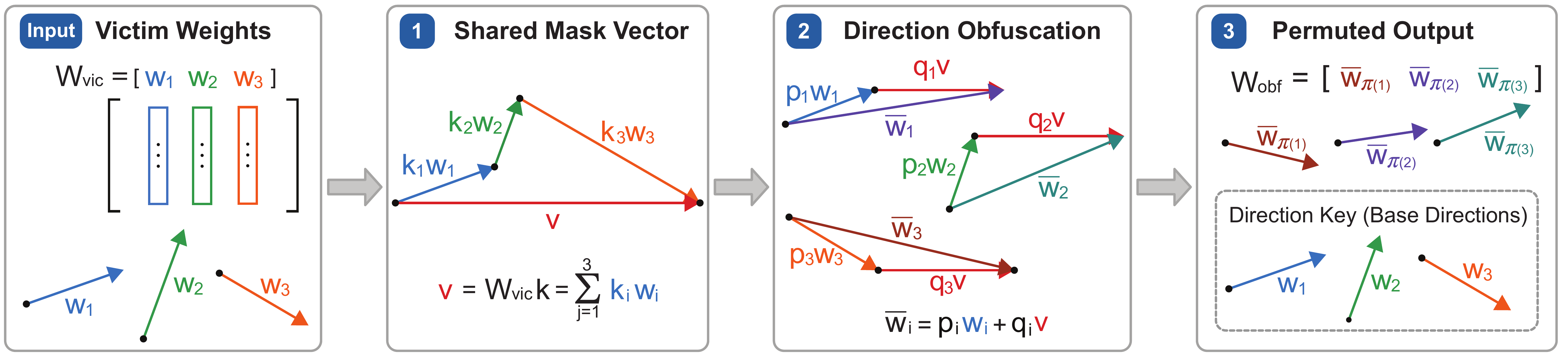}
    \caption{\arrowcloak{} counters \arrowmatch{} by adding a shared mask
    vector to rotate victim-weight vectors before permutation.}
    \label{fig:arrowcloak}
\end{figure*}

\begin{table}[!h]
\caption{Notation used throughout the paper.}
\label{tab:notation}
\centering
\footnotesize
\setlength{\tabcolsep}{3pt}
\renewcommand{\arraystretch}{1.06}
\begin{tabularx}{\columnwidth}{@{}>{\raggedright\arraybackslash}p{0.34\columnwidth}X@{}}
\toprule
\textbf{Symbol} & \textbf{Meaning} \\
\midrule
$A,W,H,\Pi$ & Matrices \\
$\bm{x},\bm{v},\bm{q}$ & Column vectors \\
$p,q,i$ & Scalars or indices \\
$\colof{A}{i},\rowof{A}{i},A\tp,A\inv$
& Column $i$, row $i$, transpose, and inverse of $A$ \\
$\Id,\ones,\stdbasis{i},\diag(\cdot),\round{\cdot}$
& Identity, ones, basis vector, diagonalization, and nearest-integer rounding \\
$\opnorm{\cdot},\frobnorm{\cdot},\cossim(\bm{x},\bm{y})$
& Euclidean/spectral norm, Frobenius norm, and cosine similarity \\
\bottomrule
\end{tabularx}
\end{table}

\subsection{Partitioned Inference with an Untrusted GPU}

A trusted execution environment (TEE) isolates selected code and data from the
rest of the system, including privileged host software~\cite{sabt2015trusted}.  This isolation can
protect proprietary model state on a user-controlled device.  Executing an
entire model inside a CPU TEE is costly because CPU TEEs lack GPU-class GEMM
throughput and incur additional trusted-memory and data-movement overhead.
Prior work therefore partitions machine-learning
workloads between a CPU TEE and a faster, untrusted accelerator, typically
retaining sensitive state and lightweight operations inside the TEE while
outsourcing expensive linear algebra~\cite{slalom,goten,threelegrace}.

TEE-Shielded LLM Partitioning (TSLP)~\cite{shen2022soter,sun2023shadownet,li2024translinkguard,xu2024tempo,
yang2024kvshield,sun2025tsqp,li2025coreguard} applies this design to on-device model
inference.  For an offloaded linear layer, the TEE
transforms the private weights and exposes only the transformed weights to the
GPU.  The GPU performs the dominant matrix multiplication, after which the TEE
recovers the output of the original layer.  Other operations that are not
suited to outsourcing remain inside the TEE.  Unlike systems that extend
trusted execution directly to the accelerator~\cite{graviton,telekine}, this
setting treats the GPU and its memory as observable by the device owner.

The weight transformation must satisfy competing security and efficiency
requirements.  It should prevent an observer of the GPU from recovering useful
private weights, while its preparation and output-recovery costs must remain
substantially below the outsourced matrix multiplication.  Consequently,
lightweight TSLP schemes commonly use vector-wise scaling, permutation, and
structured masking rather than dense secret transformations.

\subsection{From \textnormal{\arrowmatch{}} to \arrowcloak{}}
\label{sec:bg-arrow}

Among recent proposals for lightweight weight protection in TSLP,
\textit{Game of Arrows} introduces \arrowcloak{}
~\cite{gameofarrows2025}, a scheme intended to conceal
weight-vector directions while retaining inexpensive TEE-side recovery. The
scheme is developed in response to \arrowmatch{}, an attack that exposes
direction leakage in earlier lightweight obfuscation mechanisms.

Specifically, let $W_{\mathrm{pre}}$ denote the weights of a public pre-trained model and
$W_{\mathrm{vic}}$ the corresponding weights after private fine-tuning.
The attacker knows $W_{\mathrm{pre}}$ and observes the transformed weight
vectors derived from $W_{\mathrm{vic}}$ and exposed to the untrusted GPU.
\arrowmatch{} observes that corresponding weight vectors in
$W_{\mathrm{pre}}$ and $W_{\mathrm{vic}}$ often remain directionally similar. Because scaling preserves vector
direction and permutation only reorders the vectors, neither operation removes
this similarity.  An attacker can therefore compare exposed vectors with
those in $W_{\mathrm{pre}}$ and use their cosine similarities to infer the
hidden matching.  It can then reverse the remaining lightweight
transformations and construct a surrogate model~\cite{gameofarrows2025}.

\arrowcloak{} targets precisely the invariant exploited by \arrowmatch{}:
before permuting the weight vectors, it rotates each one away from its original
direction by adding a mask.  Figure~\ref{fig:arrowcloak} illustrates this
input and a three-step transformation, which we formalize below:

\noindent\textbf{Input: Victim Weights.}
Write
$W_{\mathrm{vic}}=[\bm{w}_1\;\cdots\;\bm{w}_n]
\in\mathbb{R}^{d\times n}$, where each arrow in the input panel depicts the
direction of one victim-weight column $\bm{w}_i$.

\noindent\textbf{Step~\ding{182}: Shared Mask Vector.}
\arrowcloak{} samples
$\bm{k}=(k_1,\ldots,k_n)\tp$ and forms the shared mask vector
\begin{equation}
  \bm{v}
  =W_{\mathrm{vic}}\bm{k}
  =\sum_{j=1}^{n} k_j\bm{w}_j.
  \label{eq:background-arrowcloak-mask}
\end{equation}
The red vector in step~\ding{182} is this shared mask vector, later added to
each column with a column-specific coefficient. Following \arrowcloak{}'s scheme description, each coefficient
$k_i$ takes values in $\mathbb{Z}_u$, where $u$ is a large prime.

\noindent\textbf{Step~\ding{183}: Direction Obfuscation.}
For each column, \arrowcloak{} samples a positive scale $p_i$ and a mask
coefficient $q_i$, then computes
\begin{equation}
  \overline{\bm{w}}_i
  =p_i\bm{w}_i+q_i\bm{v}.
  \label{eq:background-arrowcloak-column}
\end{equation}

The scaling term $p_i\bm{w}_i$ remains collinear with $\bm{w}_i$; adding $q_i\bm{v}$ changes the direction of the exposed vector, thereby removing the directional invariant exploited by \arrowmatch{}.

\noindent\textbf{Step~\ding{184}: Permuted Output.}
Define $D_1=\diag(p_1,\ldots,p_n)$ and
$\bm{q}=(q_1,\ldots,q_n)\tp$.  Stacking the transformed columns gives
$\overline W=W_{\mathrm{vic}}D_1+\bm{v}\bm{q}\tp
=W_{\mathrm{vic}}(D_1+\bm{k}\bm{q}\tp)$.
\arrowcloak{} then samples a permutation matrix $\Pi$ and exposes
\begin{equation}
  W_{\mathrm{obf}}
  =\overline W\Pi
  =W_{\mathrm{vic}}H\Pi,
  \qquad
  H=D_1+\bm{k}\bm{q}\tp.
  \label{eq:background-arrowcloak-transform}
\end{equation}
Here, $\Pi$ only reorders the rotated columns.

\noindent\textbf{Efficient recovery.}
For an input matrix $X$, the GPU returns
$Y_{\mathrm{obf}}=XW_{\mathrm{obf}}$.  The TEE knows $D_1$, $\bm{q}$,
$\Pi$, and $\bm{v}$ and recovers the private layer output as
\begin{equation}
  Y_{\mathrm{vic}}
  =Y_{\mathrm{obf}}\Pi\inv D_1\inv
   -(X\bm{v})\bm{q}\tp D_1\inv.
  \label{eq:background-arrowcloak-recovery}
\end{equation}
The only multiplication involving the private input and mask is the
matrix--vector product $X\bm{v}$; the remaining operations are scaling,
replication, and permutation. Accordingly, \arrowcloak{} characterizes its
online TEE cost as linear in the relevant vector dimensions, avoiding a dense
matrix--matrix multiplication.
For our later analysis, we also consider a natural rank-$r$ extension that
replaces the rank-one term $\bm{k}\bm{q}\tp$ with
$KR\tp$, thereby introducing $r$ shared mask directions.

\noindent\textbf{Claimed security argument.}
Beyond disrupting the directional similarity exploited by \arrowmatch{},
\arrowcloak{} further motivates its security through a claimed connection to
matrix Learning with Errors (LWE). It argues that, after quantization,
recovering the hidden transformation from the public and obfuscated weights
constitutes an LWE problem and is therefore computationally hard. In Section~\ref{sec:lwe-audit}, we further formalize and audit
this claim.

\section{Reassessing the LWE-Based Security Claim}
\label{sec:lwe-audit}

In this section, we formalize and reassess \arrowcloak{}'s LWE-based security. We first
review standard LWE and examine the reduction required to transfer its hardness
to \arrowcloak{} recovery. We then reconstruct \arrowcloak{}'s claimed mapping
and examine the problems in its reduction direction, arithmetic, instance distributions, and
parameter requirements.

\subsection{LWE and the Required Security Argument}
\label{sec:lwe-required-argument}

Fix a dimension $\lwedim$, a modulus $\lwemod\ge 2$, and an error
distribution $\errdist$ over $\Int$.  An LWE sample for a secret
$\bm{s}\in\Zmod{\lwemod}^{\lwedim}$ is obtained by drawing
$\bm{a}\samples\Unif(\Zmod{\lwemod}^{\lwedim})$ and $e\samples\errdist$
independently and releasing
\begin{equation}
\label{eq:lwe-sample}
(\bm{a},\,b),\qquad b=\ip{\bm{a}}{\bm{s}}+e \bmod \lwemod .
\end{equation}
Stacking $\lwesamp$ independent samples gives
$A\samples\Unif(\Zmod{\lwemod}^{\lwesamp\times\lwedim})$ and
$\bm{b}=A\bm{s}+\bm{e}\bmod\lwemod$ with $\bm{e}\samples\errdist^{\lwesamp}$. The
search-LWE problem asks to recover $\bm{s}$ from $(A,\bm{b})$~\cite{regev2009lwe,peikert2016decade}.

LWE is believed to be computationally hard for several choices of secret and
error distributions and underlies many lattice-based cryptographic
constructions~\cite{regev2009lwe,peikert2016decade}.
By contrast, Bootle et al.~\cite{bootle2018ilwe} show that the integer variant obtained by omitting
modular reduction (ILWE) is efficiently solvable under mild conditions on the
coefficient and error distributions and their relative variances, using
least-squares regression.

To base the hardness of a recovery problem $\mathcal{P}$ on LWE, one must show that an efficient solver for $\mathcal{P}$ would imply an efficient solver for LWE. 
For a direct recovery reduction, this can be done by mapping an LWE challenge to a valid instance of $\mathcal{P}$ such that solving the latter enables solving the original LWE challenge. 
Section~III-B examines whether \arrowcloak{} provides such a reduction.

\subsection{Auditing \arrowcloak{}'s LWE Mapping}
\arrowcloak{} starts from the real-valued identity
\begin{equation}
  W_{\mathrm{pre}} = W_{\mathrm{obf}}\widetilde H^{-1}
  + (W_{\mathrm{pre}}-W_{\mathrm{vic}}),
  \label{eq:arrow-lwe-map}
\end{equation}
and associates $W_{\mathrm{pre}}$ with the LWE observation,
$W_{\mathrm{obf}}$ with the public coefficient matrix,
$\widetilde H^{-1}$ with the secret, and
$E=W_{\mathrm{pre}}-W_{\mathrm{vic}}$ with the error. It
then rounds these real-valued matrices separately into a finite field and
argues that Eq.~(\ref{eq:arrow-lwe-map}) reduces \arrowcloak{} recovery to LWE.

To justify treating
$E=W_{\mathrm{pre}}-W_{\mathrm{vic}}$ as the LWE error term,
\arrowcloak{} examines its empirical distribution using weight histograms,
Kolmogorov--Smirnov tests, and its measured standard deviation
~\cite{gameofarrows2025}. These checks concern only whether the fine-tuning
residual resembles the proposed error distribution. Together, the algebraic correspondence,
quantization step, and empirical residual analysis form the basis of
\arrowcloak{}'s LWE-based security argument.

We next identify the gaps that prevent this argument from establishing the
claimed hardness.

\noindent\textbf{Reduction in the wrong direction.}
As discussed in Section~\ref{sec:lwe-required-argument}, transferring LWE
hardness requires a reduction from LWE to \arrowcloak{} recovery. Instead,
\arrowcloak{} starts from its own weight-transformation relation and rewrites it in LWE-like notation.
It provides no reduction showing how an efficient adversary that breaks \arrowcloak{} can be transformed into an algorithm that breaks the corresponding LWE hardness assumption, and therefore does not establish \arrowcloak{} recovery hardness from LWE hardness.

The absence of this reduction is already sufficient to invalidate the claimed
hardness transfer. Even with a reduction in the required direction, its output
would still need to be a valid LWE instance. The proposed mapping fails both
arithmetically and distributionally.

\noindent\textbf{Incorrect application of quantization.}
In \arrowcloak{}'s proposed mapping, the integer matrices in the claimed
LWE instance are obtained from Eq.~(\ref{eq:arrow-lwe-map}) using the
common scale $\Delta=2^m$:
\[
\begin{aligned}
A^\top &= \lfloor\Delta W_{\mathrm{obf}}\rceil,
&\qquad
S &= \lfloor\Delta\widetilde H^{-1}\rceil,\\
B &= \lfloor\Delta W_{\mathrm{pre}}\rceil,
&
E &= \lfloor\Delta(W_{\mathrm{pre}}-W_{\mathrm{vic}})\rceil.
\end{aligned}
\]
Thus, even before accounting for rounding error,
\[
A^\top S
  \approx \Delta^2 W_{\mathrm{obf}}\widetilde H^{-1},
\qquad
B-E
  \approx \Delta W_{\mathrm{obf}}\widetilde H^{-1}.
\]
Hence, $B=A^\top S+E\pmod q$ does not follow from the original
real-valued identity: the product contains an additional factor of
$\Delta$. Reduction modulo $q$ does not remove this mismatch. A valid
fixed-point encoding would need to specify a rescaling operation, such as
multiplication by $\Delta^{-1}$ in $\mathbb Z_q$, and account for the
rounding errors induced by both quantization and multiplication. The
published mapping specifies neither step.

\noindent\textbf{Mismatch with LWE distributions.}
In \arrowcloak{}, the quantized matrices $A$ and $E$ are derived from
$W_{\mathrm{vic}}H\Pi$ and $W_{\mathrm{pre}}-W_{\mathrm{vic}}$, respectively,
and therefore share dependence on the victim weights. The quantized secret $S$
is derived from $\Pi^{-1}H^{-1}$, whose structure is determined by $H$ and
$\Pi$. However, in standard matrix LWE, $A$ should be uniform, $E$ should be
sampled independently from a specified small distribution, and $S$ should
follow the secret distribution of the invoked LWE variant
~\cite{regev2009lwe,peikert2016decade}. Consequently, the published mapping
does not establish that its induced instance follows an LWE distribution.
Moreover, \arrowcloak{}'s empirical checks concern only the marginal
distribution of $E$ and therefore do not establish that the joint distribution
of $(A,S,E)$ matches the invoked LWE variant.

Taken together, these gaps prevent \arrowcloak{}'s security argument from
establishing recovery hardness based on LWE.

\section{\attackname{}: A Structural Recovery Attack}
\label{sec:recovery}
Having shown in Section~\ref{sec:lwe-audit} that \arrowcloak{}'s LWE
mapping does not establish recovery hardness, we now ask whether its concrete
construction can nevertheless be broken efficiently. Inspecting the
transformation reveals that its lightweight masking mechanism reuses a shared
direction across protected columns, leaving a rank-one structure in the exposed
weights. We exploit this structure with \attackname{}, a query-free recovery
attack that uses only the public pre-trained weights $W_{\mathrm{pre}}$ and the
exposed obfuscated weights $W_{\mathrm{obf}}$. The attack estimates and removes
the shared mask direction, recovers the hidden permutation, and reconstructs
the protected weights. We present the rank-one attack first, then extend it to
rank-$r$ masking and analyze its computational complexity.

\subsection{Threat Model and Attack Overview}
\label{sec:arrowcloak-structure}

\noindent\textbf{Threat model.}
Following the identical threat model in \arrowcloak{}~\cite{gameofarrows2025}, we assume an on-device deployment scenario and an honest-but-curious adversary model. The model user follows the
prescribed inference protocol but controls the host and untrusted GPU, and thus
observes the model architecture, tensor metadata, and obfuscated weights
$W_{\mathrm{obf}}$. Our main setting assumes access to the exact public
checkpoint $W_{\mathrm{pre}}$ from which the private victim weights
$W_{\mathrm{vic}}$ were fine-tuned. The adversary cannot access
$W_{\mathrm{vic}}$, transformation secrets, private training state, or
TEE-protected data, and does not exploit TEE vulnerabilities, side channels,
or protocol deviations. 
The adversary aims to construct a surrogate model $\widehat M$ that matches
$M_{\mathrm{vic}}$ in behavior and task performance. Exact recovery of every
victim weight or transformation secret is unnecessary.

\begin{figure*}[t]
    \centering
    \subfloat[Shared-direction leakage. The additive masks applied to
    individual columns all lie in one shared subspace.]{%
        \includegraphics[width=\textwidth]{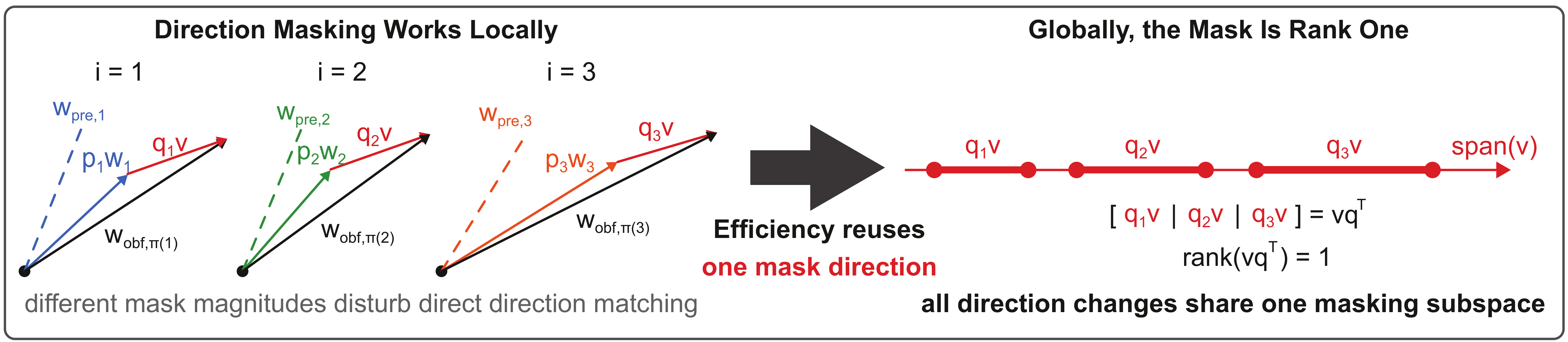}%
        \label{fig:shared-direction-leakage}}
    \par\medskip
    \subfloat[\attackname{} attack workflow. The attack estimates and removes
    the shared mask subspace, recovers the hidden column permutation through
    one-to-one assignment, and fits the remaining coefficients to reconstruct
    the victim weights.]{%
        \includegraphics[width=\textwidth]{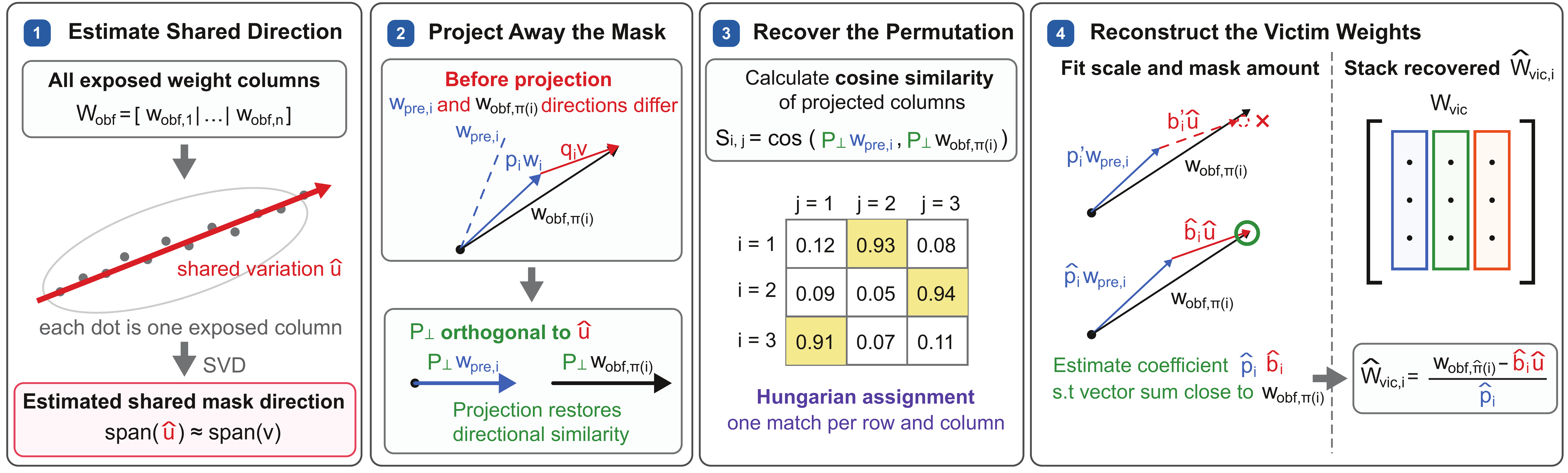}%
        \label{fig:attack-workflow}}
    \caption{Shared-direction leakage and the \attackname{} attack workflow.}
    \label{fig:attack-overview}
\end{figure*}

\noindent\textbf{Inspiration: Shared-direction leakage.} A low-rank mask is introduced in 
\arrowcloak{} to perturb the victim model weight. Recall from Section~\ref{sec:background} that it transforms a
victim weight matrix as
\begin{equation}
  W_{\mathrm{obf}}
  =(W_{\mathrm{vic}}D_1+\bm{v}\bm{q}^{\top})\Pi,
  \label{eq:rank-one-attack-model}
\end{equation}
The additive term $q_i\bm v$ changes the direction of column $i$ before permutation, but all
such terms share the same mask direction $\bm v$. Thus, for nonzero $\bm v$
and $\bm q$,
\begin{equation}
  \operatorname{rank}(\bm{v}\bm{q}^{\top}\Pi)=1,
  \qquad
  \operatorname{col}(\bm{v}\bm{q}^{\top}\Pi)
  =\operatorname{span}(\bm{v}).
  \label{eq:shared-mask-column-space}
\end{equation}
The hidden permutation $\Pi$ reorders the columns but does not conceal the shared
direction of the additive masks.

The transformation introduces an observable low-rank structure in $W_{\mathrm{obf}}$, which may be used to remove the effect of the mask direction $\bm v$.
Specifically, $\bm v$ could be estimated by SVD, in which the leading left singular direction of centered
$W_{\mathrm{obf}}$ estimates $\operatorname{span}(\bm v)$. 

Projecting out
this direction suppresses the shared mask. Because
$W_{\mathrm{vic}}=W_{\mathrm{pre}}-E$, projected public and obfuscated columns
linked by the hidden permutation remain approximately collinear when the
projected fine-tuning residual is small. Therefore, we can use
one-to-one column matching to find the hidden permutation and recover the victim model weight.

\noindent\textbf{Attack overview.}
Figure~\ref{fig:attack-overview} summarizes how \attackname{} turns this
observation into a four-stage recovery attack. First, it estimates the shared
mask direction from $W_{\mathrm{obf}}$. Second, it projects this direction out
of both the public and obfuscated weights. Third, it matches the projected
columns through a global one-to-one assignment to recover the hidden
permutation. Finally, using this permutation, it fits each column's scale and
effective mask coefficient to reconstruct the victim weights.

\subsection{Estimating the Shared Direction}
\label{sec:direction-extraction}

We now show how to estimate the shared mask direction from
$W_{\mathrm{obf}}$. Although the rank-one mask component is mixed with the scaled
victim weights, its direction can be estimated from the leading left singular
vector of the centered exposed matrix.
We first center the exposed columns by subtracting their mean. Define
\begin{equation}
  \bm\mu_{\mathrm{obf}}
  =\frac{1}{n}W_{\mathrm{obf}}\ones,
  \qquad
  \widetilde W_{\mathrm{obf}}
  =W_{\mathrm{obf}}-\bm\mu_{\mathrm{obf}}\ones\tp.
  \label{eq:attack-centered-svd}
\end{equation}
Thus, the columns of $\widetilde W_{\mathrm{obf}}$ have zero mean.
This centering makes the subsequent SVD depend on variation across columns rather than on their overall mean.

Let
\begin{equation}
  \widetilde W_{\mathrm{obf}}=U\Sigma V\tp
  \label{eq:estimated-mask-svd}
\end{equation}
be the SVD of the centered exposed matrix, with
singular values in non-increasing order. For the rank-one construction, the
attacker estimates the mask direction using the leading left singular vector,
\begin{equation}
  \widehat{\bm u}=\colof{U}{1}.
  \label{eq:estimated-mask-direction}
\end{equation}
The sign of $\widehat{\bm u}$ is immaterial because the attack uses only its
span.

\noindent\textbf{Why it works.} 
For compactness in the analysis, we write the same centering operation as
\[
\widetilde W_{\mathrm{obf}}
=W_{\mathrm{obf}}P_{\ones^\perp},
\qquad
P_{\ones^\perp}
=\Id-\frac{1}{n}\ones\ones\tp.
\]

To explain the estimator, define the following latent quantities for analysis:
\begin{equation}
\begin{aligned}
  X&=W_{\mathrm{vic}}D_1,
  &\widetilde X&=XP_{\ones^\perp}, \\
  \widetilde{\bm q}&=P_{\ones^\perp}\bm q,
  &M&=\bm v\widetilde{\bm q}\tp.
\end{aligned}
\end{equation}
Because a permutation preserves the all-ones vector,
$\Pi P_{\ones^\perp}=P_{\ones^\perp}\Pi$. Applying the centering operator to
Eq.~\eqref{eq:rank-one-attack-model} therefore gives
\begin{equation}
\begin{aligned}
  \widetilde W_{\mathrm{obf}}
  &=(X+\bm v\bm q\tp)\Pi P_{\ones^\perp} \\
  &=(XP_{\ones^\perp}+\bm v\widetilde{\bm q}\tp)\Pi \\
  &=(\widetilde X+M)\Pi.
\end{aligned}
  \label{eq:centered-observed-decomposition}
\end{equation}
We first analyze the rank-one mask signal $M$ in isolation. Each column of
$M$ is a scalar multiple of $\bm v$, so $M$ has rank at most one. If
$\bm v\neq\zerovec$ and $\widetilde{\bm q}\neq\zerovec$, at least one column
is nonzero and $M$ has rank one. To identify its leading direction, for
$\bm u=\bm v/\opnorm{\bm v}$, let
$\lambda_{\mathrm{mask}}
=\opnorm{\widetilde{\bm q}}^2\opnorm{\bm v}^2>0$.
Direct multiplication gives
\begin{equation}
\begin{aligned}
  MM\tp\bm u
  &=\opnorm{\widetilde{\bm q}}^2
    \bm v\bm v\tp\frac{\bm v}{\opnorm{\bm v}}
   =\lambda_{\mathrm{mask}}\bm u.
\end{aligned}
  \label{eq:mask-direction-rank-one}
\end{equation}
Thus, $\bm u$ is an eigenvector with positive eigenvalue
$\lambda_{\mathrm{mask}}$. Because $M$ is nonzero and rank one, $MM\tp$ is a
nonzero rank-one positive-semidefinite matrix. Therefore, $\bm u$ spans its
unique leading eigenspace, and the leading unit eigenvector is $\bm u$ up to
sign.

We now return to the observed centered matrix $\widetilde W_{\mathrm{obf}}$.
Since $\Pi$ is a permutation matrix, $\Pi\Pi\tp=\Id$. Therefore,
\begin{equation}
\begin{aligned}
  \widetilde W_{\mathrm{obf}}\widetilde W_{\mathrm{obf}}\tp
  &=(\widetilde X+M)\Pi\Pi\tp(\widetilde X+M)\tp \\
  &=(\widetilde X+M)(\widetilde X+M)\tp.
\end{aligned}
  \label{eq:permutation-invariant-gram}
\end{equation}
The leading left singular vector of $\widetilde W_{\mathrm{obf}}$ is the
leading eigenvector of this Gram matrix.

Relative to the isolated mask signal, the scaled victim component
$\widetilde X$ and the resulting cross terms act as a spectral perturbation.
When that perturbation is small relative to $\lambda_{\mathrm{mask}}$, the
leading left singular vector $\widehat{\bm u}$ approaches $\bm u$ up to sign.
Section~\ref{sec:analysis} formalizes this condition and the resulting
direction-recovery guarantee.

\subsection{Recovering the Hidden Permutation}
\label{sec:permutation-recovery}
We estimate the hidden permutation in two steps. We first show how to derive the projected
column relation under the ideal mask-removal projector. We then use the
estimated projector to construct a one-to-one assignment.

\noindent\textbf{Ideal mask removal.} We remove $\bm v$ mask by projections.
Let
\begin{equation}
  \bm u=\frac{\bm v}{\opnorm{\bm v}},
  \qquad
  P_{\bm v^\perp}=\Id-\bm u\bm u\tp.
  \label{eq:true-mask-projector}
\end{equation}
The matrix $P_{\bm v^\perp}$ is the orthogonal projector onto the hyperplane
$\bm v^{\perp}$ and satisfies $P_{\bm v^\perp}\bm v=0$.
Recall that
\begin{equation}
  W_{\mathrm{obf}}
  =
  \left(
    W_{\mathrm{vic}}D_1
    +\bm v\bm q\tp
  \right)\Pi.
  \label{eq:attack-arrow-public-low-rank}
\end{equation}
Projecting the exposed matrix onto $\bm v^{\perp}$ removes the rank-one mask:
\begin{equation}
\begin{aligned}
  P_{\bm v^\perp}W_{\mathrm{obf}}
  &=
  P_{\bm v^\perp}
  \left(
    W_{\mathrm{vic}}D_1+\bm v\bm q\tp
  \right)\Pi \\
  &=
  \left(P_{\bm v^\perp}W_{\mathrm{vic}}D_1\right)\Pi.
  \label{eq:oracle-projected-obfuscated}
\end{aligned}
\end{equation}

Because $W_{\mathrm{vic}}=W_{\mathrm{pre}}-E$, its $i$th column satisfies
$\colof{{W_{\mathrm{vic}}}}{i} = \colof{{W_{\mathrm{pre}}}}{i}-\colof{E}{i}$.
Recall that $D_1 = \operatorname{diag}(p_1, \dots, p_n)$ with positive scales $p_i > 0$.
Let $\pi(i)$ denote the exposed position assigned to victim column $i$.
Since projection removes the shared mask component, the corresponding projected column satisfies
\begin{equation}
  P_{\bm v^\perp}\colof{{W_{\mathrm{obf}}}}{\pi(i)}
  =
  p_i\left(
    P_{\bm v^\perp}\colof{{W_{\mathrm{pre}}}}{i}
    -
    P_{\bm v^\perp}\colof{E}{i}
  \right).
  \label{eq:projected-corresponding-column}
\end{equation}
Because $p_i>0$, when 
$\opnorm{P_{\bm v^\perp}\colof{E}{i}}
\ll\opnorm{P_{\bm v^\perp}\colof{{W_{\mathrm{pre}}}}{i}}$, the correct
projected pair is
nearly positively aligned and its oracle cosine similarity is close to one. This
motivates matching the projected columns by cosine similarity.

\noindent\textbf{Estimated projection and global assignment.}
In practice, the true direction $\bm u$ is unavailable. Using the estimate
$\widehat{\bm u}$ from Eq.~\eqref{eq:estimated-mask-direction}, the attacker
constructs
\begin{equation}
  P_{\widehat{\bm u}^\perp}
  =\Id-\widehat{\bm u}\widehat{\bm u}\tp.
  \label{eq:estimated-mask-projector}
\end{equation}

Let
\begin{equation}
  S_{ij}
  =
  \cossim\!\left(
    P_{\widehat{\bm u}^\perp}\colof{{W_{\mathrm{pre}}}}{i},
    P_{\widehat{\bm u}^\perp}\colof{{W_{\mathrm{obf}}}}{j}
  \right),
  \label{eq:projected-cosine-score}
\end{equation}
be the projected cosine similarity between public column $i$ and exposed
column $j$. Since $\pi$ is bijective, every public column must be assigned to
exactly one exposed column, and every exposed column must be used exactly once.
We therefore solve
\begin{equation}
  \widehat\pi
  =
  \arg\max_{\sigma\in\mathcal S_n}
  \sum_{i=1}^{n}S_{i,\sigma(i)}.
  \label{eq:attack-global-assignment}
\end{equation}
We solve this one-to-one linear-assignment problem using the Hungarian
algorithm~\cite{kuhn1955hungarian} applied to $-S$. The
Hungarian algorithm enforces the one-to-one constraint globally, yielding a
valid permutation.

Using $P_{\widehat{\bm u}^\perp}$ instead of $P_{\bm v^\perp}$ introduces an
additional score perturbation beyond the fine-tuning residual in
Eq.~\eqref{eq:projected-corresponding-column}.
Section~\ref{sec:projector-recovery-guarantee} bounds the projector error, and
Appendix~\ref{app:matching-guarantee} shows that the assignment exactly
recovers $\pi$ when the projected columns retain nonzero norm and the correct
cosine matches have sufficient margin.

\subsection{Reconstructing the Victim Weights}
Given the recovered permutation and mask direction, we estimate each
column's scale and effective mask coefficient by constrained least squares and
then invert the fitted transformation.

\noindent\textbf{Coefficient estimation.}
For a correctly matched column, such that $\widehat\pi(i)=\pi(i)$, the exposed
column satisfies
\begin{equation}
\begin{aligned}
  \colof{{W_{\mathrm{obf}}}}{\widehat\pi(i)}
  &=p_i\colof{{W_{\mathrm{vic}}}}{i}+q_i\bm v \\
  &=p_i\colof{{W_{\mathrm{pre}}}}{i}
    +q_i\bm v-p_i\colof{E}{i}.
\end{aligned}
  \label{eq:matched-column-decomposition}
\end{equation}
Because $\widehat{\bm u}\approx\pm\bm v/\opnorm{\bm v}$ and the fine-tuning
residual is small, we use the approximation
\begin{equation}
  \colof{{W_{\mathrm{obf}}}}{\widehat\pi(i)}
  \approx
  p_i\colof{{W_{\mathrm{pre}}}}{i}+b_i\widehat{\bm u},
  \label{eq:rank-one-two-component-model}
\end{equation}
where $b_i$ absorbs the mask coefficient, the norm of $\bm v$, and the sign
ambiguity of $\widehat{\bm u}$. We estimate the two coefficients by constrained
least squares:
\begin{equation}
  (\widehat p_i,\widehat b_i)
  =
  \arg\min_{p\in[p_{\min},p_{\max}],\,b\in\mathbb{R}}
  \opnorm{
    \colof{{W_{\mathrm{obf}}}}{\widehat\pi(i)}
    -p\colof{{W_{\mathrm{pre}}}}{i}
    -b\widehat{\bm u}
  }^2.
  \label{eq:rank-one-coefficient-fit}
\end{equation}

The interval $[p_{\min},p_{\max}]$ enforces the scaling range used by
\arrowcloak{}, while $b$ is unconstrained because it absorbs the unknown
mask magnitude and sign.

\noindent\textbf{Weight reconstruction.}
Using the fitted coefficients, we reconstruct the victim column as
\begin{equation}
  \colof{\widehat{W_{\mathrm{vic}}}}{i}
  =\frac{\colof{{W_{\mathrm{obf}}}}{\widehat\pi(i)}
          -\widehat b_i\widehat{\bm u}}
         {\widehat p_i}.
  \label{eq:rank-one-victim-reconstruction}
\end{equation}
Although the public column is used to estimate the two coefficients, the
reconstruction uses the original exposed column. The fine-tuning residual in
Eq.~\eqref{eq:matched-column-decomposition} is therefore retained after mask
removal rather than replaced by zero. The remaining reconstruction error
reflects permutation-recovery errors, scale-estimation error, and residual mask
error after subtracting $\widehat b_i\widehat{\bm u}$.

\subsection{Rank-$r$ Extension}
\label{sec:rank-k-extension}
The preceding subsections focus on the rank-one case, but our attack extends
naturally to a rank-$r$ mask by replacing the single shared direction with an
$r$-dimensional mask subspace.

\noindent\textbf{Rank-$r$ mask structure.}
Following the rank-$r$ construction introduced in
Section~\ref{sec:background}, the additive mask
$(W_{\mathrm{vic}}K)R\tp$ has rank at most $r$ and admits a factorization
$BC\tp$, where $B\in\mathbb{R}^{d\times r}$ has orthonormal columns and
$C\in\mathbb{R}^{n\times r}$ contains the corresponding coefficients. Thus,
\begin{equation}
\begin{aligned}
  \overline W
  &=W_{\mathrm{vic}}D_1+(W_{\mathrm{vic}}K)R\tp \\
  &=W_{\mathrm{vic}}D_1+BC\tp, \\
  W_{\mathrm{obf}}
  &=\overline W\Pi.
\end{aligned}
  \label{eq:rank-k-attack-model}
\end{equation}
Thus, the columns of $B$ form an orthonormal basis for the shared mask
subspace, and the rows of $C$ contain the column-specific mask coefficients.
As in the rank-one case, the permutation only reorders the exposed columns and
does not alter the shared mask subspace.

Let $\widetilde X=W_{\mathrm{vic}}D_1P_{\ones^\perp}$ and
$\widetilde C=P_{\ones^\perp}C$. Applying the same column centering as in
Eq.~\eqref{eq:attack-centered-svd} gives
\begin{equation}
  \widetilde W_{\mathrm{obf}}
  =
  \left(
    \widetilde X+B\widetilde C\tp
  \right)\Pi.
  \label{eq:rank-k-centered-model}
\end{equation}

\noindent\textbf{Automatic rank inference.}
When the mask rank is not available from the scheme configuration, the
attacker estimates it from the exposed spectrum. Let $s_j$ denote the $j$th
singular value in the SVD of $\widetilde W_{\mathrm{obf}}$. The estimator
examines the adjacent singular-value ratios
\begin{equation}
g_j = \frac{s_j}{s_{j+1}},
\qquad
\widehat{r} = \arg\max_j g_j .
\label{eq:spectral-ratio}
\end{equation}
Let \(g_{(1)}\) and \(g_{(2)}\) denote the largest and second-largest
ratios, respectively. The estimator accepts \(\widehat{r}\) only if
\begin{equation}
g_{(1)} \geq \tau_{\mathrm{gap}}
\qquad\text{and}\qquad
\frac{g_{(1)}}{g_{(2)}} \geq \tau_{\mathrm{sep}},
\label{eq:spectral-acceptance}
\end{equation}
where \(\tau_{\mathrm{gap}}\) controls the minimum break magnitude and
\(\tau_{\mathrm{sep}}\) controls its separation from competing breaks.

\noindent\textbf{Mask-subspace estimation and removal.}
Instead of using the single estimated direction $\widehat{\bm u}$, the
attacker uses the SVD in Eq.~\eqref{eq:estimated-mask-svd} to define
\begin{equation}
  \widehat B
  =
  \left[
    \colof{U}{1}\mid\cdots\mid\colof{U}{\widehat r}
  \right]
  \in\mathbb{R}^{d\times\widehat r}.
  \label{eq:rank-k-estimated-basis}
\end{equation}
The columns of $\widehat B$ are orthonormal and span an estimate of
$\operatorname{col}(B)$.
It then replaces the rank-one projector with
\begin{equation}
  P_{\spn(\widehat B)^\perp}
  =
  \Id-\widehat B\widehat B\tp.
  \label{eq:rank-k-estimated-projector}
\end{equation}

\noindent\textbf{Matching and reconstruction.}
Matching proceeds on the projected columns as in
Section~\ref{sec:permutation-recovery}. After matching, the attack replaces
the scalar term
$b\widehat{\bm u}$ with $\widehat B\bm c$ and fits
\begin{equation}
  (\widehat p_i,\widehat{\bm c}_i)
  =
  \arg\min_{\substack{
    p\in[p_{\min},p_{\max}],\\
    \bm c\in\mathbb{R}^{\widehat r}}}
  \opnorm{
    \colof{{W_{\mathrm{obf}}}}{\widehat\pi(i)}
    -p\colof{{W_{\mathrm{pre}}}}{i}
    -\widehat B\bm c
  }^2.
  \label{eq:rank-k-coefficient-fit}
\end{equation}
The reconstructed victim column is
\begin{equation}
  \colof{\widehat{W_{\mathrm{vic}}}}{i}
  =
  \frac{
    \colof{{W_{\mathrm{obf}}}}{\widehat\pi(i)}
    -\widehat B\widehat{\bm c}_i
  }{\widehat p_i}.
  \label{eq:rank-k-victim-reconstruction}
\end{equation}
Thus, the matching and reconstruction logic is unchanged, but successful
recovery additionally requires an accurate rank and mask-subspace estimate.

\subsection{Computational Complexity}

For a $d$-by-$n$ weight matrix, the attack has three dominant computational
costs: dense SVD, pairwise projected cosine-similarity construction, and
one-to-one assignment. Computing the dense SVD of the centered exposed matrix
takes
\begin{equation}
  O\left(dn\min\{d,n\}\right)
\end{equation}
time. Constructing the projected cosine-similarity matrix requires
$O(n^2d)$ time, and solving the resulting one-to-one assignment with the
Hungarian algorithm requires $O(n^3)$ time.

Projection and per-column coefficient fitting require
$O(dn\widehat r)$ and $O(nd\widehat r^2)$ time, respectively. For the small
mask ranks targeted by the scheme, these terms are lower order than the dense
similarity construction and assignment. The overall runtime is therefore
\begin{equation}
  O\left(
    dn\min\{d,n\}
    +n^2d
    +n^3
  \right).
\end{equation}
The overall space complexity is $O(dn+n^2)$ for the input, projected, and
similarity matrices.

\section{Recovery Guarantees and Failure Conditions}
\label{sec:analysis}

A key step in the attack presented in Section~\ref{sec:recovery} is to
estimate the shared mask direction from the exposed weights. This estimate
underpins the subsequent mask-removal, column-matching, and victim-weight
reconstruction stages. In this section, we analyze when and how accurately
the shared direction can be estimated. We first establish sufficient
conditions under which the rank-one mask signal dominates the residual
spectrum, and then translate this spectral separation into bounds on the
estimated direction and the resulting projection operator. Finally, we
extend the analysis to rank-$r$ masking and discuss the corresponding failure
conditions. Throughout this section, quantities unavailable to the attacker,
including $\widetilde X$, $\bm v$, and $\widetilde{\bm q}$, are introduced only
to state and analyze the guarantees; the attack itself continues to use only
$W_{\mathrm{pre}}$ and $W_{\mathrm{obf}}$.

\subsection{Dominance of the Rank-One Mask Signal}

To characterize when the shared mask direction can be accurately estimated,
we compare the rank-one mask signal with the remaining terms in the Gram
matrix of the centered exposed weights. Recall from
Eq.~\eqref{eq:centered-observed-decomposition} that
\begin{equation}
  \widetilde W_{\mathrm{obf}}
  =
  (\widetilde X+\bm v\widetilde{\bm q}\tp)\Pi.
\end{equation}
Because $\Pi\Pi\tp=\Id$,
\begin{equation}
\begin{aligned}
  G
  &:={}
  \widetilde W_{\mathrm{obf}}
  \widetilde W_{\mathrm{obf}}\tp \\
  &={}
  \lambda_{\mathrm{mask}}\bm u\bm u\tp+R,
  \label{eq:analysis-gram}
\end{aligned}
\end{equation}
where
\begin{equation}
  \bm u=\frac{\bm v}{\opnorm{\bm v}},
  \qquad
  \lambda_{\mathrm{mask}}
  =\opnorm{\widetilde{\bm q}}^2\opnorm{\bm v}^2,
  \label{eq:analysis-mask-strength}
\end{equation}
and
\begin{equation}
  R
  =
  \widetilde X\widetilde X\tp
  +\widetilde X\widetilde{\bm q}\bm v\tp
  +\bm v\widetilde{\bm q}\tp\widetilde X\tp.
  \label{eq:analysis-residual}
\end{equation}
Let $\delta=\opnorm{R}$.

\begin{center}
\begin{minipage}[t]{1\columnwidth}
\vspace{-8pt}
\begin{mdframed}
\textbf{Proposition (Rank-One Spectral Dominance).}
Suppose the centered mask coefficients have constant-order empirical variance,
$\opnorm{\bm v}^2=\Omega(n)$
(that is, $\bm v$ aggregates a non-vanishing fraction of victim columns
without catastrophic cancellation),
and the scaled victim columns have bounded average squared norm. Then
\[
\begin{aligned}
\lambda_{\mathrm{mask}}&=\Omega(n^2),
&
\frac{\delta}{\lambda_{\mathrm{mask}}}&=O(n^{-1/2}).
\end{aligned}
\]
In particular, a superlinear lower bound on $\opnorm{\bm v}^2$ only
strengthens the same ratio.
\end{mdframed}
\end{minipage}
\end{center}

\begin{IEEEproof}
The first condition gives
\begin{equation}
  \opnorm{\widetilde{\bm q}}^2
  =\sum_{i=1}^{n}(q_i-\bar q)^2
  =\Theta(n).
  \label{eq:analysis-centered-mask-scaling}
\end{equation}
Combined with the hypothesis $\opnorm{\bm v}^2=\Omega(n)$ and
Eq.~\eqref{eq:analysis-mask-strength},
\begin{equation}
  \lambda_{\mathrm{mask}}
  =\opnorm{\widetilde{\bm q}}^2\opnorm{\bm v}^2
  =\Omega(n^2).
  \label{eq:analysis-mask-scaling}
\end{equation}

For the scaled victim component,
\begin{equation}
  \frobnorm{X}^2
  =\sum_{i=1}^{n}\opnorm{p_i\bm w_i}^2
  =O(n).
  \label{eq:analysis-scaled-victim-energy}
\end{equation}
Since column centering is an orthogonal projection,
\begin{equation}
\begin{aligned}
  \opnorm{\widetilde X}
  &\leq\frobnorm{\widetilde X}
  \leq\frobnorm{X}
  =O(\sqrt n).
\end{aligned}
  \label{eq:analysis-structural-scaling}
\end{equation}
By submultiplicativity,
\begin{equation}
\begin{aligned}
  \delta
  &\leq
  \opnorm{\widetilde X}^2
  +2\opnorm{\widetilde X}
    \opnorm{\widetilde{\bm q}}\opnorm{\bm v}.
\end{aligned}
  \label{eq:analysis-residual-scaling}
\end{equation}
Dividing by $\lambda_{\mathrm{mask}}=\opnorm{\widetilde{\bm q}}^2\opnorm{\bm v}^2$
yields
\begin{equation}
\begin{aligned}
  \frac{\delta}{\lambda_{\mathrm{mask}}}
  &\leq
  \frac{\opnorm{\widetilde X}^2}
       {\opnorm{\widetilde{\bm q}}^2\opnorm{\bm v}^2}
  +
  \frac{2\opnorm{\widetilde X}}
       {\opnorm{\widetilde{\bm q}}\opnorm{\bm v}}.
\end{aligned}
\end{equation}
The first term is $O(n)/(\Theta(n)\cdot\Omega(n))=O(n^{-1})$.
The second term is $O(\sqrt n)/(\Theta(\sqrt n)\cdot\Omega(\sqrt n))=O(n^{-1/2})$.
Hence
\begin{equation}
  \frac{\delta}{\lambda_{\mathrm{mask}}}
  =O(n^{-1/2}).
  \label{eq:analysis-spectral-ratio}
\end{equation}
\end{IEEEproof}

The $O(n^{-1/2})$ rate predicts stronger spectral separation for wider
protected matrices. Because the bound is asymptotic and conditional, we
measure $\delta/\lambda_{\mathrm{mask}}$ directly for each protected layer in
Section~\ref{sec:evaluation}.

\subsection{Recovery of the Mask Direction}

The preceding proposition bounds the residual relative to the rank-one mask
signal. We now translate this spectral separation into an angular error bound
for the estimated mask direction. Let $\widehat{\bm u}$ be the unit leading
left singular vector of $\widetilde W_{\mathrm{obf}}$, equivalently a leading
eigenvector of $G$. Define its angle from the true mask direction, up to sign,
as
\begin{equation}
  \theta
  =
  \arccos\!\left(\left|\widehat{\bm u}\tp\bm u\right|\right)
  \in[0,\pi/2].
  \label{eq:analysis-direction-angle}
\end{equation}

\begin{center}
\begin{minipage}[t]{1\columnwidth}
\vspace{-8pt}
\begin{mdframed}
\textbf{Proposition (Mask-Direction Recovery).}
Under the decomposition in Eq.~\eqref{eq:analysis-gram}, the angular error of
the leading left singular direction satisfies
\[
  \sin^2\theta
  \leq
  \frac{2\delta}{\lambda_{\mathrm{mask}}}.
\]
\end{mdframed}
\end{minipage}
\end{center}

\begin{IEEEproof}
Because $G=\widetilde W_{\mathrm{obf}}\widetilde W_{\mathrm{obf}}\tp$ is a
symmetric positive-semidefinite Gram matrix, the Rayleigh--Ritz
characterization implies that its unit leading eigenvector $\widehat{\bm u}$
maximizes the Rayleigh quotient. Because $\bm u$ is also a unit vector,
\begin{equation}
  \widehat{\bm u}\tp G\widehat{\bm u}
  =\lambda_{\max}(G)
  \geq
  \bm u\tp G\bm u.
\end{equation}
Substituting Eq.~\eqref{eq:analysis-gram} and using
$\opnorm{\bm u}=1$ gives
\begin{equation}
\begin{aligned}
  &\lambda_{\mathrm{mask}}
  (\widehat{\bm u}\tp\bm u)^2
  +\widehat{\bm u}\tp R\widehat{\bm u} \\
  &\qquad\geq
  \lambda_{\mathrm{mask}}+\bm u\tp R\bm u.
\end{aligned}
\end{equation}
Rearranging and applying
$|\bm x\tp R\bm x|\leq\opnorm{R}=\delta$ for every unit vector $\bm x$
gives
\begin{equation}
\begin{aligned}
  \lambda_{\mathrm{mask}}
  \left(1-(\widehat{\bm u}\tp\bm u)^2\right)
  &\leq
  \widehat{\bm u}\tp R\widehat{\bm u}
  -\bm u\tp R\bm u \\
  &\leq 2\delta,
\end{aligned}
  \label{eq:analysis-direction-derivation}
\end{equation}
Since $\sin^2\theta=1-(\widehat{\bm u}\tp\bm u)^2$,
\begin{equation}
  \sin^2\theta
  \leq
  \frac{2\delta}{\lambda_{\mathrm{mask}}}.
  \label{eq:analysis-direction-bound}
\end{equation}
\end{IEEEproof}

Combining this bound with Eq.~\eqref{eq:analysis-spectral-ratio} gives
$\sin\theta=O(n^{-1/4})$. Thus, under the stated scaling conditions, wider
protected matrices admit a tighter direction-recovery bound.

\subsection{Recovery of the Mask-Removal Projector}
\label{sec:projector-recovery-guarantee}

The attack uses the estimated mask direction through its orthogonal-complement
projector. We therefore translate the angular error in $\widehat{\bm u}$ into
an operator-norm error for the resulting mask-removal projection. Following
Section~\ref{sec:permutation-recovery}, define
\begin{equation}
  P_{\bm v^\perp}=\Id-\bm u\bm u\tp,
  \qquad
  P_{\widehat{\bm u}^\perp}
  =\Id-\widehat{\bm u}\widehat{\bm u}\tp.
  \label{eq:analysis-projectors}
\end{equation}

\begin{center}
\begin{minipage}[t]{1\columnwidth}
\vspace{-8pt}
\begin{mdframed}
\textbf{Proposition (Mask-Removal Projector Recovery).}
The true and estimated mask-removal projectors satisfy
\[
\begin{aligned}
  \opnorm{P_{\widehat{\bm u}^\perp}-P_{\bm v^\perp}}
  &=\sin\theta\\
  &\leq\sqrt{\frac{2\delta}{\lambda_{\mathrm{mask}}}}.
\end{aligned}
\]
\end{mdframed}
\end{minipage}
\end{center}

\begin{IEEEproof}
Because $P_{\widehat{\bm u}^\perp}$ is invariant to the sign of
$\widehat{\bm u}$, choose the sign so that
$c=\widehat{\bm u}\tp\bm u=\cos\theta\geq0$, and let
$s=\sin\theta=\sqrt{1-c^2}$. If $s=0$, then
$\widehat{\bm u}=\bm u$ under this sign choice, so the two projectors are
identical. Otherwise, define
\begin{equation}
  \bm z
  =\frac{\widehat{\bm u}-c\bm u}{s}.
\end{equation}
Since
\begin{equation}
  \bm u\tp(\widehat{\bm u}-c\bm u)=c-c=0,
  \qquad
  \opnorm{\widehat{\bm u}-c\bm u}^2=1-c^2=s^2,
\end{equation}
$\bm z$ is a unit vector orthogonal to $\bm u$, and
$\widehat{\bm u}=c\bm u+s\bm z$. In the orthonormal basis
$\{\bm u,\bm z\}$,
\begin{equation}
\begin{aligned}
  P_{\widehat{\bm u}^\perp}-P_{\bm v^\perp}
  &=\bm u\bm u\tp-\widehat{\bm u}\widehat{\bm u}\tp\\
  &=
  \begin{bmatrix}
    1 & 0 \\
    0 & 0
  \end{bmatrix}
  -
  \begin{bmatrix}
    c^2 & cs \\
    cs & s^2
  \end{bmatrix}\\
  &=
  \begin{bmatrix}
    s^2 & -cs \\
    -cs & -s^2
  \end{bmatrix}.
\end{aligned}
  \label{eq:analysis-projector-matrix}
\end{equation}
This matrix has eigenvalues $s$ and $-s$ and is zero outside
$\operatorname{span}\{\bm u,\bm z\}$. Hence,
\begin{equation}
\begin{aligned}
  \opnorm{P_{\widehat{\bm u}^\perp}-P_{\bm v^\perp}}
  &=s
  =\sin\theta \\
  &\leq
  \sqrt{\frac{2\delta}{\lambda_{\mathrm{mask}}}}.
  \label{eq:analysis-projector-bound}
\end{aligned}
\end{equation}
where the inequality follows from Eq.~\eqref{eq:analysis-direction-bound}.
\end{IEEEproof}

Consequently, for any weight column $\bm x$,
\begin{equation}
  \opnorm{
    P_{\widehat{\bm u}^\perp}\bm x
    -P_{\bm v^\perp}\bm x
  }
  \leq
  \sqrt{\frac{2\delta}{\lambda_{\mathrm{mask}}}}
  \opnorm{\bm x}.
  \label{eq:analysis-projector-column-error}
\end{equation}
This bound controls the additional column-wise error introduced by the
estimated projector. Exact permutation recovery additionally requires
nonzero projected column norms and a positive matching margin, as formalized
in Appendix~\ref{app:matching-guarantee}.

\subsection{Rank-$r$ Extension and Failure Conditions}
\label{sec:rank-k-guarantee}

For a centered rank-$r$ mask $B\widetilde C^{\top}$, where the columns of
$B$ are orthonormal, the mask contribution to the Gram matrix is
\begin{equation}
  B(\widetilde C^{\top}\widetilde C)B^{\top}.
  \label{eq:analysis-rank-k-signal}
\end{equation}
Its nonzero eigenspace is $\operatorname{span}(B)$. The proof follows the
same steps as the rank-one case: replace the mask direction by the top-$r$
mask subspace and compare its weakest nonzero eigenvalue with the residual
spectrum. When this eigengap dominates the residual, the estimated and true
subspace projectors are close. 

Across both the rank-one and rank-$r$ settings, recovery becomes unreliable
when the mask signal is not sufficiently separated from the residual or when
the projected columns lose their
separation. Rank-$r$ masking introduces additional failure modes: the rank may
be misestimated, the weakest mask direction may be too weak, or $r$ may be so
large that projection removes most of the matching information.

\section{Evaluation}
\label{sec:evaluation}

We evaluate \attackname{} as an end-to-end attack against \arrowcloak{}.
This section focuses on the three research questions that establish the
attack's effectiveness, mechanism, and implications for \arrowcloak{}'s
security premise:

\begin{itemize}
    \item \textbf{RQ1: End-to-End Recovery.}
    Can \attackname{} recover the hidden permutation and reproduce victim
    functionality using only the public pre-trained weights and the
    obfuscated weights exposed outside the TEE?

    \item \textbf{RQ2: Recovery Mechanism.}
    Does the shared mask expose a recoverable low-rank subspace, and is the
    estimated mask-removal projector accurate enough to enable the
    permutation recovery underlying the end-to-end attack in RQ1?

    \item \textbf{RQ3: Error-Distribution Security Premise.}
    Do real fine-tuning residuals satisfy the normality and
    standard-deviation conditions used to support \arrowcloak{}'s claimed LWE
    mapping, and does \attackname{} remain effective when those conditions hold?
\end{itemize}

We also report the attacker-side runtime in this section.
Additionally, we answer \textbf{RQ4: Higher-Rank Masking as
Mitigation} in
Section~\ref{sec:rq4-mitigation}, where we examine how increasing the mask
rank affects structural permutation recovery and functional victim
recovery. We report \textbf{RQ5: Robustness to Attacker Uncertainty} last as
a supplementary evaluation in Appendix~\ref{app:attacker-uncertainty}, since
it concerns how an attacker can infer the mask rank and select a compatible
public reference rather than the core recovery result.

\subsection{Experimental Setup}

\noindent\textbf{Models and workloads.}
We evaluate six model--task pairs: ViT-B/16~\cite{dosovitskiy2020image} on CIFAR-100~\cite{krizhevsky2009learning}, BERT-base~\cite{devlin2019bert} on
SST-2~\cite{socher2013recursive}, GPT-2~\cite{radford2019language} on SST-2, ViT on Kvasir-SEG~\cite{jha2019kvasir}, and SD~2.1~\cite{rombach2022high} LoRA on the Jewelry
and Emoji datasets. We fine-tune ViT-B/16 on CIFAR-100 for 10 epochs
(3,910 optimization steps) with an initial learning rate of $10^{-3}$.
BERT-base and GPT-2 are fine-tuned on SST-2 for three epochs
(6,315 steps) with initial learning rates of $2\times10^{-5}$ and
$10^{-5}$, respectively. The Kvasir-SEG model is trained for 12 epochs
(1,200 steps), using learning rates of $3\times10^{-4}$ for the
segmentation decoder and $3\times10^{-5}$ for the ViT backbone.
For SD~2.1, we use the released Jewelry LoRA trained for 500 steps with
a learning rate of $10^{-4}$ and the released Emoji LoRA checkpoint at
1,000 steps, whose initial learning rate is $10^{-4}$.
For each pair, $W_{\mathrm{pre}}$ denotes the public checkpoint used to
initialize fine-tuning, and $W_{\mathrm{vic}}$ denotes the resulting
victim checkpoint. We apply \arrowcloak{} independently to each protected
linear layer and record the corresponding $W_{\mathrm{obf}}$ exposed to
the GPU.

\noindent\textbf{Defense and attack settings.}
Unless otherwise stated, we instantiate \arrowcloak{}'s original rank-one
construction using the exact coefficient distributions in its
released implementation. The transformation
secrets are resampled independently for each protected layer.

For each protected matrix, \attackname{} receives only
$W_{\mathrm{pre}}$ and $W_{\mathrm{obf}}$. The victim weights and
\arrowcloak{}'s transformation secrets are used only for evaluating recovery
and are never provided to the attack. We keep all attack hyperparameters
fixed across comparable experiments. Rank-$r$ experiments vary only the
defense-side mask rank. 

We run all experiments on an AMD machine, with EPYC 9115 16-core CPU and an NVIDIA RTX 6000 GPU. 

\subsection{RQ1: End-to-End Recovery}

We first evaluate whether the structural recovery achieved by \attackname{}
translates into an end-to-end compromise of real protected models. We apply
the complete attack to six protected-model configurations spanning
classification, segmentation, and diffusion.

We use permutation accuracy as the primary structural endpoint. It measures
the fraction of protected weight vectors assigned to their ground-truth
positions by the final one-to-one assignment. We additionally report mean
true-match rank. For each protected reference vector, we rank all exposed
candidates by their matching scores and record the position of the
ground-truth match. A mean rank close to one indicates that the true
match remains among the highest-scoring candidates even when exact recovery
fails.

\begin{table}[!t]
\caption{Structural recovery across six rank-one protected-model
configurations.}
\label{tab:rq1-structural}
\centering
\begingroup
\setlength{\tabcolsep}{1.8pt}
\renewcommand{\arraystretch}{1.08}
\scriptsize
\begin{tabular*}{\columnwidth}{@{\extracolsep{\fill}}lrr@{}}
\toprule
Configuration & Perm. acc. (\%) $\uparrow$ & Mean true-match rank $\downarrow$ \\
\midrule
ViT-B/16/CIFAR-100                 &  99.95 & 1.27   \\
BERT-base/SST-2                    & 100.00 & 1.00   \\
GPT-2/SST-2                        & 100.00 & 1.00   \\
ViT/Kvasir-SEG                     &  99.92 & 1.31   \\
SD 2.1/Jewelry                     & 100.00 & 1.0034 \\
SD 2.1/Emoji                       & 100.00 & 1.0000 \\
\bottomrule
\end{tabular*}
\endgroup
\end{table}

Table~\ref{tab:rq1-structural} shows that \attackname{}
recovers the hidden permutation almost perfectly across all six
configurations, with accuracy ranging from 99.92\% to 100\%.
Mean true-match rank ranges from 1.00 to 1.31, showing that the ground-truth
counterparts remain near the top of the candidate rankings even in the few
cases where the final assignment is incorrect.

Having established structural recovery, we next examine whether the
reconstructed weights reproduce the victim models' functionality. For the
three classification models, we report victim accuracy, recovered-model
accuracy, and prediction agreement with the victim. For Kvasir-SEG, we report
the victim and recovered models' Dice and IoU against the ground-truth
segmentation masks, together with agreement between their outputs. All three
metrics are reported as percentages, with higher values indicating better
performance or agreement. The two diffusion configurations do not provide a
directly comparable task-performance benchmark. We therefore evaluate only
victim--recovered output similarity, using PSNR and SSIM on generations
produced with matched prompts and random seeds to measure pixel-level fidelity
and structural similarity, respectively.

\begin{figure}[t]
    \centering
    \includegraphics[width=\linewidth]{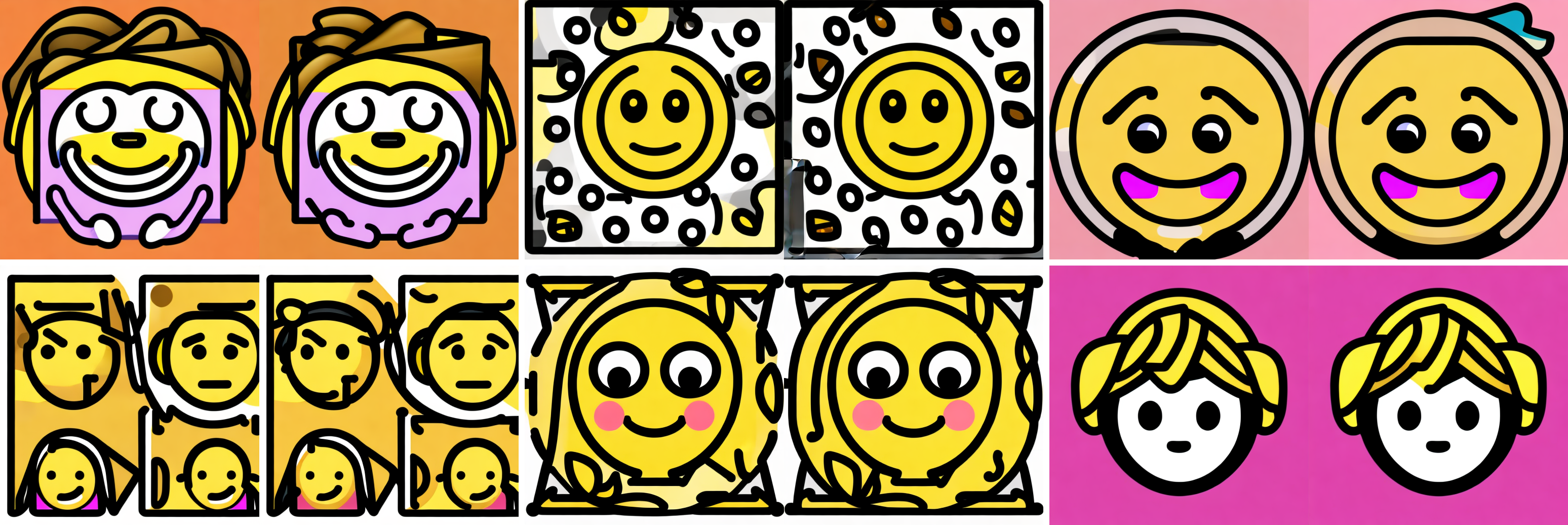}
    \caption{Qualitative victim--recovered comparison on SD 2.1/Emoji under
    matched prompts and random seeds. Within each pair, the victim output is
    shown on the left and the recovered-model output on the right.}
    \label{fig:emoji-qualitative}
\end{figure}

\begin{table}[!t]
\caption{Functional recovery for classification and segmentation. We report
task performance for the victim and recovered models together with their
output agreement.}
\label{tab:rq1-functional}
\centering
\setlength{\tabcolsep}{1.6pt}
\renewcommand{\arraystretch}{1.08}
\scriptsize
\begin{tabular*}{\columnwidth}{@{\extracolsep{\fill}}lrrr@{}}
\toprule
Configuration (metric) & Victim & Rec. & Agree. (\%) \\
\midrule
ViT/CIFAR-100 (Acc.)       & 92.37       & 90.78       & 94.39 \\
BERT/SST-2 (Acc.)          & 91.63       & 91.86       & 99.54 \\
GPT-2/SST-2 (Acc.)         & 90.83       & 91.06       & 97.94 \\
Kvasir-SEG (Dice/IoU)      & 89.61/83.00 & 87.58/80.12 & 98.35 \\
\bottomrule
\end{tabular*}
\end{table}

Table~\ref{tab:rq1-functional} shows that classification
accuracy differs from the victim by at most
1.59 percentage points, while prediction agreement ranges from 94.39\% to
99.54\%. On Kvasir-SEG, the recovered model is within 2.03 percentage points
in Dice and 2.88 percentage points in IoU of the victim, with 98.35\% output
agreement. For diffusion, matched victim--recovered generations yield a PSNR
of 18.5224 and an SSIM of 0.6125 on Jewelry, and a PSNR of 14.4303 and an SSIM
of 0.8530 on Emoji.

Figure~\ref{fig:emoji-qualitative} provides a qualitative view of the Emoji results. 
Across the six examples, the recovered model preserves the main emoji subject. 
Differences remain in fine-grained details such as accessories, background patterns, and exact contour placement. 
Such discrepancies can arise from small residual recovery errors that propagate and may compound over successive denoising steps.
Nevertheless, the strong structural similarity demonstrates the effectiveness of our recovery in reproducing the victim model's main visual characteristics.

\begin{rqresultbox}{1}
\attackname{} achieves near-perfect permutation recovery across all six
configurations, completes the structural-to-functional recovery chain for
classification and segmentation, and quantifies victim--recovered output
similarity for diffusion.
\end{rqresultbox}

\subsection{RQ2: Recovery Mechanism}

We next examine why the end-to-end recovery in RQ1 succeeds.
Section~\ref{sec:analysis} predicts that shared-mask reuse produces a
spectrally dominant rank-one component, causing the residual-to-mask ratio
$\delta/\lambda_{\mathrm{mask}}$ to decrease as $O(n^{-1/2})$ under the
stated non-cancellation and bounded-energy conditions. We test whether this
spectral separation materializes in real protected layers, whether
\attackname{} accurately recovers the resulting mask subspace, and whether
the estimated mask-removal projector is sufficiently accurate to support the
Hungarian matching stage used in RQ1.

We evaluate this recovery chain independently on every protected matrix of
ViT-B/16, BERT-base, and GPT-2. We measure three successive quantities in the
recovery chain.
First, the residual-to-mask spectral ratio
$\delta/\lambda_{\mathrm{mask}}$ measures whether the shared-mask signal
dominates the remaining matrix structure. Second, the principal angle
$\theta(\widehat{\bm u},\bm u)$ measures how accurately the leading singular
direction recovers the true mask span. Third, defining
$P=P_{\bm v^\perp}$ and $\widehat P=P_{\widehat{\bm u}^\perp}$,
$\|\widehat P-P\|_2$ measures the error in the mask-removal operation used
before permutation recovery.

\begin{table}[!t]
\caption{Realized accuracy of shared-mask subspace recovery. Results are
reported as median (maximum) across protected layers.}
\label{tab:analysis-validation}
\centering
\begingroup
\setlength{\tabcolsep}{2.3pt}
\renewcommand{\arraystretch}{1.12}
\scriptsize
\begin{tabular*}{\columnwidth}{@{\extracolsep{\fill}}lccc@{}}
\toprule
Model/task &
\shortstack{Spectral ratio $\downarrow$\\
$\delta/\lambda_{\mathrm{mask}}\;(\times10^{-4})$} &
\shortstack{Direction error $\downarrow$\\
$\theta(\widehat u,u)\;({}^\circ)$} &
\shortstack{Projector error $\downarrow$\\
$\|\widehat P-P\|_2\;(\times10^{-4})$} \\
\midrule
ViT-B/16/CIFAR-100 & 2.88 (5.41) & 0.0157 (0.0283) & 2.75 (4.95) \\
BERT-base/SST-2    & 2.92 (5.17) & 0.0160 (0.0287) & 2.79 (5.01) \\
GPT-2/SST-2        & 2.84 (5.90) & 0.0154 (0.0319) & 2.69 (5.57) \\
\bottomrule
\end{tabular*}
\endgroup
\end{table}

Table~\ref{tab:analysis-validation} shows that the shared-mask component is
sharply separated from the residual structure. The residual-to-mask ratio is
below $6\times10^{-4}$ even in the worst layer. Consistent with the analysis,
the recovered direction deviates from the true mask direction by at most
$0.032^\circ$, and the estimated and oracle mask-removal projectors differ by
at most $5.6\times10^{-4}$ in operator norm.

Appendix~\ref{app:finite-width-diagnostics} separately verifies that the
evaluated layers exhibit neither substantial cancellation in the shared mask
vector nor growth in average victim-column energy, consistent with the
finite-width premises of our analysis.

\begin{rqresultbox}{2}
    Shared-mask reuse creates a stable, spectrally identifiable subspace.
\attackname{} therefore closely approximates the oracle mask-removal
projector, producing projected similarities that support near-perfect global
permutation recovery.
\end{rqresultbox}

\subsection{RQ3: Error-Distribution Security Premise}
To justify treating the fine-tuning residual
$E=W_{\mathrm{pre}}-W_{\mathrm{vic}}$ as the LWE error term, \arrowcloak{}
empirically checks whether $E$ is approximately Gaussian and satisfies a
stated standard-deviation threshold. We evaluate these error-distribution
conditions in two complementary settings. First, for real victim residuals,
we test raw and quantized normality and apply the stated standard-deviation
threshold after quantization. Second, we construct a controlled Gaussian-$E$
victim to test whether satisfying these conditions prevents permutation
recovery.

\noindent\textbf{Real victim residuals.}
\arrowcloak{} assumes that $E$ follows a Gaussian distribution and uses its
quantized form in the claimed LWE mapping. We therefore examine the residual
both before and after quantization. 

\begin{table}[!t]
\caption{Error-distribution checks on real fine-tuning residuals.}
\label{tab:real-e-checks}
\centering
\begingroup
\setlength{\tabcolsep}{1.2pt}
\renewcommand{\arraystretch}{1.08}
\scriptsize
\begin{tabular*}{\columnwidth}{@{\extracolsep{\fill}}lrrrr@{}}
\toprule
Model/task & Layers & \shortstack{Raw\\normal} &
\shortstack{Quant. std\\threshold} & \shortstack{Quant.\\normal} \\
\midrule
ViT-B/16/CIFAR-100 & 48 & 1/48 & 48/48 & 1/48 \\
BERT-base/SST-2 & 73 & 73/73 & 73/73 & 73/73 \\
GPT-2/SST-2 & 48 & 39/48 & 48/48 & 39/48 \\
ViT/Kvasir-SEG & 48 & 6/48 & 48/48 & 5/48 \\
SD 2.1/Jewelry & 128 & 0/128 & 128/128 & 0/128 \\
SD 2.1/Emoji & 128 & 0/128 & 128/128 & 0/128 \\
\bottomrule
\end{tabular*}
\endgroup
\end{table}

Table~\ref{tab:real-e-checks} reports how many protected layers pass each
error-distribution check. Raw normal applies \arrowcloak{}'s normality check
to the real-valued residual. Quant. normal tests the Gaussian-shape requirement  after
quantization, while Quant. std threshold tests whether the quantized residual
satisfies the stated condition $\sigma>2\sqrt{n}$.

All protected layers satisfy the quantized standard-deviation
threshold, but the normality results vary substantially across models. Raw
normality holds in many BERT and GPT-2 layers but few layers in other models. BERT
satisfies both quantized checks in all 73 layers (72 Transformer linear matrices and the SST-2 classification head), and GPT-2 does so in 39 of
48 layers. In contrast, the remaining real victims satisfy quantized
normality in few or no protected layers. Thus, the error-distribution
conditions are not consistently satisfied across real victims.

\noindent\textbf{Controlled Gaussian-$E$ victim.}
We construct a controlled victim from ViT-B/16 on CIFAR-100. For each of
its 48 layers, we sample a Gaussian residual $E$ and set
$W_{\mathrm{vic}}=W_{\mathrm{pre}}-E$. This experiment tests whether the
Gaussian and variance properties of $E$ alone prevent permutation recovery.

All 48 synthetic residuals satisfy raw normality and the stated
standard-deviation threshold, while 42 of 48 satisfy both quantized checks.
Nevertheless, \attackname{} achieves 99.914\% permutation recovery, with a
mean true-match rank of 1.2507. Thus, constructing
the victim with a Gaussian residual does not remove the shared structure
exploited by the attack. RQ1 provides a corresponding real-victim example in
BERT: all 73 protected layers satisfy both quantized conditions, yet
\attackname{} achieves 100\% permutation recovery and 99.54\% prediction
agreement.

\begin{rqresultbox}{3}
The error-distribution conditions invoked by \arrowcloak{} are neither
consistently satisfied by real victims nor sufficient to prevent
recovery by \attackname{} when they are satisfied.
\end{rqresultbox}

\subsection{Attack Cost}

We measure the end-to-end runtime of the complete attack on BERT-base with
SST-2 while varying the mask rank. The attack
takes 11.531 seconds at $r=1$ and 17.170 seconds at $r=32$, with intermediate
settings falling within this range. These results show that the attack remains
computationally practical across the evaluated mask ranks.

\section{Security Implications and Mitigation}
\label{sec:lessons}

The evaluation shows that \arrowcloak{}'s original rank-one masking remains vulnerable: \attackname{} achieves near-perfect permutation recovery and closely reproduces victim functionality. 
This observation reveals a shift from local to cross-vector leakage and
motivates a natural mitigation: increasing the mask rank.
We first discuss the broader security implications of this shift and
then evaluate the resulting security--efficiency tradeoff.

\subsection{Local and Cross-Vector Leakage}
\label{sec:security-implications}

Our attack, \attackname{}, recovers the hidden permutation and reconstructs
victim functionality even after \arrowcloak{} disrupts the per-vector
directional similarity exploited by \arrowmatch{}. This result complements
\textit{Game of Arrows} rather than extending the same leakage mechanism:
\arrowmatch{} exploits local geometry preserved within individual vectors by
scaling and permutation, whereas \attackname{} exploits cross-vector structure
created when \arrowcloak{} reuses a shared mask direction, across the released vectors.

Together, the two attacks show that concealing one local invariant does not
secure a joint release if the protection mechanism introduces dependencies
across released values. Confidential-offloading defenses should therefore be
evaluated under the complete accelerator view, including public initialization
checkpoints, per-vector invariants, and dependencies across
jointly released weights. Designs that amortize trusted work by reusing low-dimensional secret
state should be tested against adaptive attacks that estimate or
remove that state.

\subsection{The Security--Efficiency Tradeoff of Higher-Rank Masking}
\label{sec:rq4-mitigation}

\arrowcloak{} independently scales and permutes the victim-weight vectors and
adds a mask to each vector. Although these masks disrupt the pairwise
directional similarity exploited by \arrowmatch{}, they are not independent
across the released vectors. In the original construction, every mask lies in
one shared direction; in the rank-$r$ extension, the masks lie in a shared
$r$-dimensional subspace. This reuse creates a coherent matrix-level signal
that can be estimated and removed when the subspace is spectrally visible.

Increasing the mask rank is thus a natural mitigation. 
A larger $r$ leaves fewer complementary dimensions for matching and places a larger fraction of each victim-weight vector inside the ambiguous mask subspace. 
We therefore test whether increasing $r$ can successfully disrupt structural recovery, functional recovery, or both.

\bheading{Mitigation effectiveness.}
We evaluate high mask rank on BERT-base with SST-2, whose masking space has dimension 768. 
We vary the mask rank from 64 to 512 and apply the automatic rank estimator and the complete attack at each rank. 
We report the estimated rank, permutation accuracy, and prediction agreement with the victim.

\begin{table}[!t]
\caption{Structural and functional recovery under higher-rank masking on
BERT-base with SST-2; permutation accuracy and agreement are in percent.}
\label{tab:higher-rank-mitigation}
\centering
\begingroup
\renewcommand{\arraystretch}{1.08}
\footnotesize
\begin{tabular*}{\columnwidth}{
  @{\extracolsep{\fill}}
  rrrr
  @{}
}
\toprule
True $r$ & Est. $\widehat{r}$ & Perm. acc. & Agreement \\
\midrule
64  & 64  & 100.00 & 95.53 \\
128 & 128 & 100.00 & 90.02 \\
256 & 256 & 100.00 & 74.43 \\
512 & 512 &  99.98 & 48.05 \\
\bottomrule
\end{tabular*}
\endgroup
\end{table}

Table~\ref{tab:higher-rank-mitigation} shows that increasing the mask rank primarily degrades functional reconstruction rather than rank estimation or permutation matching. 
The automatic estimator returns the true rank at every tested setting, while permutation accuracy remains 100\% through $r=256$ and 99.98\% at $r=512$.
This persistence follows from the coherent, spectrally identifiable structure of the rank-$r$ mask. 
SVD recovers its dominant directions, and projecting them out leaves sufficient directional information for column matching. 
Even at $r=512$, the complementary subspace retains 256 dimensions.

Functional recovery responds differently. 
Agreement with the victim decreases from 95.53\% at $r=64$ to 48.05\% at $r=512$, compared with 99.54\% for the original rank-one configuration. 
Within the estimated mask subspace, the victim component is confounded with the additive mask. 
As $r$ increases, this ambiguous component occupies a larger fraction of each weight vector, and the coefficient fit cannot recover a growing portion of the victim
weights. 
Higher rank therefore leaves the structural correspondence observable but prevents it from supporting accurate functional reconstruction.

\bheading{TEE computation overhead.}
The original rank-one design requires the TEE to apply one mask direction
during each protected linear operation. A rank-$r$ construction applies $r$
directions, so the rank-dependent part of the TEE correction requires $r$
times the arithmetic of its rank-one counterpart.

We normalize this cost against a no-offloading baseline in which all 72
protected BERT-base matrix multiplications execute inside the TEE. Under the
original rank-one design, the correction requires only 0.20\% of this
baseline. For each protected matrix, the rank-$r$ correction scales with $r$
times the sum of its input and output dimensions, whereas full TEE execution
scales with their product. Aggregating these operation counts over the 72
protected matrices gives a correction-to-full-TEE ratio of $r/512$.

At $r=512$, the setting that reduces victim agreement to 48.05\%, the TEE
correction therefore requires approximately as much arithmetic as executing
all protected matrix multiplications inside the TEE. This comparison reflects
theoretical operation counts rather than measured latency, which depends on
the TEE hardware and implementation. High-rank masking can thus suppress
functional leakage, but at $r=512$ it does so at the cost of increasing trusted
computation inside TEE.

\section{Related Work}
\label{sec:related}

TEE-assisted inference systems place the GPU on one of two sides of the trust boundary. 
Graviton, Telekine, and NVIDIA Confidential Computing on H100 GPUs place GPU execution inside the trusted boundary and provide hardware-backed isolation and attestation~\cite{graviton,telekine,nvidia_h100_cc}. 
CPU-TEE partitioning systems instead treat the GPU as untrusted: expensive linear operations run on the accelerator, whereas secrets and selected operations remain inside the CPU TEE~\cite{slalom}. 
Our attack targets the information exposed by this latter design choice, not the security of the TEE itself.

Systems in this setting differ in what they protect and how they transform inputs or model weights before outsourcing computation. 
Slalom protects private inputs and verifies outsourced linear operations~\cite{slalom}.  
SOTER additionally protects model parameters through parameter morphing and checks the untrusted execution using oblivious fingerprints~\cite{shen2022soter}.  
ShadowNet transforms the weights of linear layers before exposing them to an accelerator, restores the corresponding outputs inside the TEE, and retains nonlinear layers within the trusted boundary~\cite{sun2023shadownet}.
These designs share a central efficiency constraint: weight transformation and output recovery on the trusted CPU must remain substantially cheaper than the matrix multiplication delegated to the GPU.

The closest prior security analysis is \textit{Game of Arrows}~\cite{gameofarrows2025}, which shows that lightweight scaling and permutation preserve the directions of protected weight vectors.
Its \arrowmatch{} attack uses a public pre-trained model to recover the hidden correspondence from this directional similarity and then uses query-labeled data to adjust the recovered vector lengths.  
The same work introduces \arrowcloak{}, which adds a shared mask direction to disrupt direct direction matching.  

A recent concurrent preprint, \textit{MOSAIC}, questions \arrowcloak{}'s LWE-based security argument by observing that replacing the fine-tuning residual with fresh independent Gaussian noise produces distinguishable outputs. 
This critique targets the reduction rather than the concrete scheme and does not by itself provide a practical weight- or key-recovery attack~\cite{chiang2026mosaic}.
Our attack targets \arrowcloak{} itself. \attackname{} estimates and removes the shared low-rank mask subspace before recovering permutation. It reconstructs a functional surrogate without query-labeled data or knowledge of the secret transformation parameters. The analysis also extends from the published rank-one construction to rank-$r$ masks.

\section{Conclusion}

In this paper, we reassessed the LWE-based security claim of \arrowcloak{} and showed that
its LWE-shaped reformulation does not transfer standard LWE hardness to model
recovery. Examining the
concrete construction, we found that reusing a shared mask subspace leaves a
low-rank signal unchanged by the hidden permutation. We exploit this leakage
with \attackname{}, a polynomial-time attack that uses only the public
pre-trained and exposed obfuscated weights to estimate and remove the mask
subspace, recover the hidden permutation, and reconstruct the
victim weights. We establish sufficient spectral-separation conditions for
accurate recovery of the mask-removal projector. Our evaluation validates
this leakage: across six protected-model configurations, \attackname{}
recovers 99.92--100\% of the hidden permutation and produces surrogate
models that closely reproduce victim functionality.

\section*{Ethics Considerations}
This work examines the confidentiality guarantees of \arrowcloak{} and
identifies a vulnerability that could enable the unauthorized recovery of
proprietary models. Our experiments were conducted locally using public
checkpoints and datasets and the released \arrowcloak{} implementation; we did
not target deployed third-party systems, access private user data, or exploit
TEE vulnerabilities. We report these findings to support independent security
evaluation and the development of stronger confidential-inference defenses.

\bibliographystyle{IEEEtran}
\bibliography{reference}

\appendices

\section{Matching Stability and Permutation Recovery}
\label{app:matching-guarantee}

Section~\ref{sec:analysis} bounds the error of the estimated mask-removal
projector. We now show when this error is small enough to preserve the
projected cosine similarities and the hidden one-to-one correspondence.

Let $\mathcal W$ contain all nonzero public and exposed columns used in
matching, and define
\begin{equation}
  \eta
  =
  \|\widehat P-P\|_2,
  \qquad
  \rho
  =
  \min_{w\in\mathcal W}
  \frac{\|Pw\|_2}{\|w\|_2}.
  \label{eq:appendix-eta-rho}
\end{equation}
The quantity $\rho$ measures the smallest fraction of column norm retained
after ideal mask removal. A positive $\rho$ excludes columns that lie almost
entirely in the mask direction.

\begin{lemma}[Normalized projection stability]
\label{lem:normalized-projection-stability}
Assume $\rho>0$ and $\eta<\rho$. Then, for every
$w\in\mathcal W$,
\begin{equation}
  \left\|
    \frac{\widehat Pw}{\|\widehat Pw\|_2}
    -
    \frac{Pw}{\|Pw\|_2}
  \right\|_2
  \leq
  \frac{2\eta}{\rho}.
  \label{eq:appendix-normalized-projection}
\end{equation}
\end{lemma}

\begin{IEEEproof}
Let $a=Pw$ and $\widehat a=\widehat Pw$. By
Eq.~\eqref{eq:appendix-eta-rho},
\begin{equation}
  \|\widehat a-a\|_2
  \leq
  \eta\|w\|_2,
  \qquad
  \|a\|_2
  \geq
  \rho\|w\|_2.
  \label{eq:appendix-projection-perturbation}
\end{equation}
Moreover,
\begin{equation}
  \|\widehat a\|_2
  \geq
  (\rho-\eta)\|w\|_2
  >0.
\end{equation}
Thus both normalized vectors are well defined. The reverse triangle
inequality gives
\begin{equation}
\begin{aligned}
  \left\|
    \frac{\widehat a}{\|\widehat a\|_2}
    -
    \frac{a}{\|a\|_2}
  \right\|_2
  &\leq
  \frac{
    \big|\|a\|_2-\|\widehat a\|_2\big|
  }{\|a\|_2}
  +
  \frac{\|\widehat a-a\|_2}{\|a\|_2} \\
  &\leq
  \frac{2\|\widehat a-a\|_2}{\|a\|_2}
  \leq
  \frac{2\eta}{\rho}.
\end{aligned}
\end{equation}
\end{IEEEproof}

Define the oracle cosine score
\begin{equation}
  S^\star_{ij}
  =
  \operatorname{cos}
  \left(
    Pw_{\mathrm{pre},i},
    Pw_{\mathrm{obf},j}
  \right),
  \label{eq:appendix-oracle-score}
\end{equation}
which uses the unavailable true projector. Let $S_{ij}$ be the estimated
score in Eq.~\eqref{eq:projected-cosine-score}. Applying
Lemma~\ref{lem:normalized-projection-stability} to both normalized arguments
of the cosine and adding and subtracting their mixed inner product gives
\begin{equation}
  |S_{ij}-S^\star_{ij}|
  \leq
  \frac{4\eta}{\rho}.
  \label{eq:appendix-score-perturbation}
\end{equation}
For public column $i$, define
\begin{equation}
  \epsilon_i
  =
  \max_j|S_{ij}-S^\star_{ij}|
  \leq
  \frac{4\eta}{\rho},
  \label{eq:appendix-score-error}
\end{equation}
and the oracle matching margin
\begin{equation}
  \mu_i
  =
  S^\star_{i,\pi(i)}
  -
  \max_{j\neq\pi(i)}
  S^\star_{ij}.
  \label{eq:appendix-matching-margin}
\end{equation}

\begin{theorem}[Permutation recovery]
\label{thm:permutation-recovery}
If $\mu_i>2\epsilon_i$ for every $i$, then the linear assignment in
Eq.~\eqref{eq:attack-global-assignment} uniquely recovers the hidden
permutation $\pi$.
\end{theorem}

\begin{IEEEproof}
For every $j\neq\pi(i)$,
\begin{equation}
\begin{aligned}
  S_{i,\pi(i)}-S_{ij}
  &\geq
  S^\star_{i,\pi(i)}-S^\star_{ij}-2\epsilon_i \\
  &\geq
  \mu_i-2\epsilon_i
  >0.
\end{aligned}
\end{equation}
Thus, the correct exposed column is the unique highest-scoring match in every
row of $S$. Because $\pi$ is bijective, it is the unique permutation that
simultaneously selects all row-wise maxima. The Hungarian algorithm therefore
returns $\pi$.
\end{IEEEproof}

Combining Eq.~\eqref{eq:appendix-score-error} with
Eq.~\eqref{eq:analysis-projector-bound} gives the explicit sufficient
condition
\begin{equation}
  \mu_i
  >
  \frac{8}{\rho}
  \sqrt{\frac{2\delta}{\lambda_{\mathrm{mask}}}}
  \qquad
  \text{for every }i.
  \label{eq:appendix-combined-condition}
\end{equation}
This condition separates the two requirements for exact matching: the
estimated mask projector must be accurate, and the ideal projected columns
must retain enough norm and pairwise separation.

\section{Finite-Width Checks for Spectral Dominance}
\label{app:finite-width-diagnostics}

Section~\ref{sec:analysis} predicts spectral dominance when the shared mask
vector combines victim columns without substantial cancellation and the
scaled victim columns have bounded average squared norm. These are sufficient
conditions used to explain the realized recovery measurements in
Table~\ref{tab:analysis-validation}, rather than attack-success metrics by
themselves. We examine whether they are consistent with the protected
matrices at their deployed widths.

For each protected matrix, we draw \arrowcloak{} coefficients
$k_i\sim\operatorname{Unif}\{0,1,2,3,4\}$, form
$\bm v=W_{\mathrm{vic}}\bm k$, and record
\begin{equation}
  r_v=\frac{\opnorm{\bm v}^2}{n},
  \qquad
  c_v=
  \frac{\opnorm{\bm v}^2}
       {\sum_{i=1}^{n}k_i^2\opnorm{\bm w_i}^2}.
  \label{eq:finite-width-diagnostics}
\end{equation}
The normalized energy $r_v$ checks whether $\opnorm{\bm v}^2$ remains on the
order of $n$ at the evaluated widths. The cancellation ratio $c_v$ compares
the realized energy of the sum with the sum of the individual contribution
energies; a value near zero would indicate that the victim columns largely
cancel one another. We also measure $\|X\|_F^2/n$ for
$X=W_{\mathrm{vic}}D_1$, which is the average squared norm of the scaled
victim columns.

\begin{table}[!t]
\caption{Finite-width checks for the spectral-dominance premises.
$r_v$ and $c_v$ are the median (minimum) of per-layer medians over 100
defense draws; $\|X\|_F^2/n$ is the median (maximum) across layers.}
\label{tab:omega-v}
\centering
\begingroup
\setlength{\tabcolsep}{4pt}
\renewcommand{\arraystretch}{1.12}
\scriptsize
\begin{tabular*}{\columnwidth}{@{\extracolsep{\fill}}lccc@{}}
\toprule
Model/task &
\shortstack{$r_v$ $\uparrow$\\
med (min)} &
\shortstack{$c_v$ $\uparrow$\\
med (min)} &
\shortstack{$\|X\|_F^2/n$\\
med (max)} \\
\midrule
ViT-B/16/CIFAR-100 & 2.49 (0.56) & 0.99 (0.33) & 0.98 (6.29) \\
BERT-base/SST-2    & 8.86 (1.36) & 0.99 (0.36) & 3.19 (12.22) \\
GPT-2/SST-2        & 92.30 (13.15) & 0.96 (0.35) & 32.58 (301.87) \\
\bottomrule
\end{tabular*}
\endgroup
\end{table}

Table~\ref{tab:omega-v} shows that the shared direction retains
non-negligible energy at the evaluated widths. Across all protected matrices,
the minimum observed $r_v$ is 0.56 and the minimum $c_v$ is 0.33, providing
no evidence of severe cancellation. The centered mask coefficients have
empirical variance approximately $10$, consistent with their constant-order
variance in the analysis. Although GPT-2 has a larger absolute column scale,
the within-backbone Spearman correlations between $\|X\|_F^2/n$ and $n$ are
$-0.27$, $-0.13$, and $-0.25$ for ViT, BERT, and GPT-2, respectively; thus,
we observe no increase with width in these finite models. These measurements
are consistent with the non-degeneracy premises at the deployed widths, but
do not by themselves prove the asymptotic $\Omega(n)$ and $O(n)$ conditions.

\section{Supplementary Evaluation}
\label{app:attacker-uncertainty}

\subsection{RQ5: Robustness to Attacker Uncertainty}

The rank-$r$ extension automatically estimates the mask rank from the exposed
spectrum, while correspondence matching requires a compatible public
reference. RQ5 evaluates rank inference and reference selection using only
the exposed weights and public information.

\noindent\textbf{Unknown mask rank.}
We evaluate automatic rank inference on BERT-base with SST-2, whose masking
space has dimension 768. For this evaluation, we instantiate the acceptance
thresholds as $\tau_{\mathrm{gap}}=10$ and
$\tau_{\mathrm{sep}}=1.5$. We vary the true mask rank over
$r\in\{2,4,\ldots,512\}$ and apply the automatic estimator described in
Section~\ref{sec:rank-k-extension} independently to each protected layer.

\begin{table}[!t]
\caption{Automatic rank inference and end-to-end recovery on BERT-base with
SST-2; accuracies are in percent.}
\label{tab:rank-k-recovery}
\centering
\begingroup
\renewcommand{\arraystretch}{1.07}
\footnotesize

\setlength{\tabcolsep}{3.2pt}
\begin{tabular*}{\columnwidth}{@{\extracolsep{\fill}}rrrr@{}}
\toprule
\multicolumn{4}{c}{\textit{Panel A: Rank-inference confidence}} \\
\midrule
\shortstack{True\\$r$}
& \shortstack{$\widehat{r}=r$\\layers}
& \shortstack{Minimum\\$g_{(1)}$}
& \shortstack{Minimum\\$g_{(1)}/g_{(2)}$} \\
\midrule
2   & 72/72 & 346.85 & 28.53 \\
4   & 72/72 & 379.03 & 25.62 \\
8   & 72/72 & 345.77 & 17.87 \\
16  & 72/72 & 341.39 & 15.03 \\
32  & 72/72 & 314.70 & 12.35 \\
64  & 72/72 & 306.22 & 11.40 \\
128 & 72/72 & 301.06 &  9.24 \\
256 & 72/72 & 232.27 &  7.51 \\
512 & 72/72 &  95.27 &  3.73 \\
\bottomrule
\end{tabular*}

\vspace{7pt}

\begin{tabular*}{\columnwidth}{@{\extracolsep{\fill}}rrrr@{}}
\toprule
\multicolumn{4}{c}{\textit{Panel B: Recovery with the inferred rank}} \\
\midrule
\shortstack{True\\$r$}
& $\widehat{r}$
& \shortstack{Perm.\\acc.}
& \shortstack{Recover\\acc.} \\
\midrule
2  & 2  & 100.00 & 91.51 \\
4  & 4  & 100.00 & 91.74 \\
8  & 8  & 100.00 & 92.20 \\
16 & 16 & 100.00 & 91.28 \\
32 & 32 & 100.00 & 91.97 \\
64 & 64 & 100.00 & 90.37 \\
\bottomrule
\end{tabular*}

\endgroup
\end{table}

Panel~A of Table~\ref{tab:rank-k-recovery} shows that the estimator recovers
the true rank in all 72 protected layers for every tested $r$. Even at
$r=512$, where the mask occupies two thirds of the 768-dimensional space,
the weakest layer has $g_{(1)}=95.27$ and
$g_{(1)}/g_{(2)}=3.73$, well above the respective acceptance thresholds.
Panel~B evaluates the complete attack using these inferred ranks. It achieves
100\% permutation accuracy for every end-to-end setting from $r=2$ to
$r=64$, while recovered accuracy remains between 90.37\% and 92.20\%.
Thus, the attack does not require oracle knowledge of the mask rank in the
evaluated settings. RQ4 separately examines the defensive effect of
increasing $r$ on functional recovery.

\noindent\textbf{Unknown public reference.}
The exposed tensor dimensions and GPU-side operations reveal the protected
layer shapes, allowing the attacker to narrow the search to
architecture-compatible public checkpoints. We evaluate reference selection
using GPT-2 as a case study. For each candidate, we automatically estimate
and remove the masking subspace, and then compare every projected exposed
column with all columns in the candidate reference. Let $s_{l,i}^{(1)}$ and
$s_{l,i}^{(2)}$ denote the largest and second-largest cosine similarities for
column $i$ in layer $l$. We define the reference-selection score as
\begin{equation}
  S=
  \frac{
    \sum_l\sum_{i=1}^{n_l}
    \left(s_{l,i}^{(1)}-s_{l,i}^{(2)}\right)
  }{
    \sum_l n_l
  },
  \label{eq:reference-selection-score}
\end{equation}
where $n_l$ is the number of protected columns in layer $l$. A larger score
indicates clearer and more consistent correspondences across the protected
layers.

\begin{table}[!t]
\caption{Public-reference selection among architecture-compatible GPT-2
checkpoints; accuracies are in percent.}
\label{tab:public-reference-mismatch}
\centering
\begingroup
\setlength{\tabcolsep}{1.4pt}
\renewcommand{\arraystretch}{1.08}
\scriptsize
\begin{tabular*}{\columnwidth}{@{\extracolsep{\fill}}llrrr@{}}
\toprule
Candidate & Relation &
\shortstack{Selection\\score} &
\shortstack{Perm.\\acc.} &
\shortstack{Recover\\acc.} \\
\midrule
Exact GPT-2
& \shortstack[l]{Victim\\initialization}
& 0.80472 & 100.00 & 91.06 \\

\shortstack[l]{Mission\\NoShuffle}
& \shortstack[l]{Independent\\pre-training}
& 0.00921 & 0.09 & 51.83 \\
\bottomrule
\end{tabular*}
\endgroup
\end{table}

The exact GPT-2 initialization receives a selection score of 0.80472 and
enables perfect permutation recovery. In contrast, the independently
pretrained checkpoint receives a score of
0.00921, yields only 0.09\% permutation accuracy, and reduces recovered
accuracy to 51.83\%. The selection score therefore separates compatible and
incompatible protected backbones in this candidate set.

\begin{rqresultbox}{5}
\attackname{} is robust to both forms of uncertainty when the exposed spectrum
contains a clear rank break and the candidate pool contains the correct
protected backbone. The mask rank need not be known in advance. However, an
independently pretrained checkpoint with the same architecture is not a
sufficient substitute when the correct protected backbone is unavailable.
\end{rqresultbox}

\end{document}